\documentclass{article}
\usepackage{authblk}
\usepackage{graphicx}
\usepackage{amssymb,latexsym,amsmath,bbm,graphicx,amsthm}
\usepackage[all]{xypic}
\usepackage{multicol}
\usepackage{url,color}

\newenvironment{tbs}{%
   \small\tt
   \renewcommand{\labelitemi}{$\blacktriangleright$}%
   \begin{itemize}}{\end{itemize}}
\newcommand{\btbs}{\begin{tbs}}                                                                      
\newcommand{\etbs}{\end{tbs}}

\newcommand{\hide}[1]{}

\newtheorem{theo}{Theorem}
\newtheorem{prop}{Proposition}

\theoremstyle{definition}
\newtheorem{defi}{Definition}
\theoremstyle{definition}
\newtheorem{rema}{Remark}
\theoremstyle{definition}
\newtheorem{exam}{Example}
\theoremstyle{definition}
\newtheorem{claim}{Claim}

\newcommand{\funP}{\mathsf{P}}

\newcommand{\Prop}{\ensuremath{\mathtt{X}}}        %PROP

\renewcommand{\phi}{\varphi} % nicer \phi

\newcommand{\bbS}{\mathfrak{M}}

\title{Flat modal fixpoint logics with the converse modality}
\author{Sebastian Enqvist \\ 
\small
Email: thesebastianenqvist@gmail.com \\
Phone: +46 (0)761903819
\normalsize
}
\affil{Department of Philosophy, Stockholm University}

\date{\today}

\begin{document}

\maketitle
\begin{abstract}
We prove a generic completeness result for a class of  modal fixpoint logics corresponding to flat fragments of the two-way mu-calculus, extending earlier work by Santocanale and Venema. We observe that Santocanale and Venema's proof that least fixpoints in the Lindenbaum-Tarski algebra of certain flat fixpoint logics are constructive, using finitary adjoints, no longer works when the converse modality is introduced. Instead, our completeness proof directly constructs a model for a consistent formula, using the induction rule in a way that is similar to the standard completeness proof for propositional dynamic logic. This approach is combined with the concept of a focus, which has previously been used in tableau based reasoning for modal fixpoint logics.  
\end{abstract}

%%%%%%%%%%         TEXT          %%%%%%%%%%%%%%%%%%%%%%%%%%%%%
%
%\btbs
%\item GOAL: prove completeness for all extensions of K4n, S5n and the basic temporal language (next-time and its converse) by untied fixpoint connectives, and give a simpler proof for the case of K as a corollary to the proof.
%\item Note: the problem for K4n does not follow by translation into K with transitive closure, since this translation might introduce alternation!
%\item Problem: the method of $\mathcal{O}$-adjoints due to Venema and Santocanale is less straightforward for free algebras of modal logics above K. (Maybe doesn't work at all.)
%\etbs

%\input{approach-via-tightening}
%\input{adjoints-for-S5}
%\input{sec-motivation}
%\input{sec-model-construction}
\newcommand{\fdia}{\Diamond^{\scriptscriptstyle\mathsf{F}}}
\newcommand{\bdia}{\Diamond^{\scriptscriptstyle\mathsf{B}}}
\renewcommand{\fbox}{\Box^{\scriptscriptstyle\mathsf{F}}}
\newcommand{\bbox}{\Box^{\scriptscriptstyle\mathsf{B}}}
\newcommand{\fnabla}{\nabla^{\scriptscriptstyle\mathsf{F}}}
\newcommand{\bnabla}{\nabla^{\scriptscriptstyle\mathsf{B}}}
\newcommand{\edge}[1]{\stackrel{#1}{\longrightarrow}}
\newcommand{\gpath}[1]{\;{{\stackrel{{#1}*}{\longrightarrow}}}\;}
\newcommand{\reach}[1]{\;{{\stackrel{{#1}*}{\longrightarrow}}}\;}
\newcommand{\equp}[1]{\equiv^\uparrow_{#1}}
\newcommand{\simup}[3]{#1 \setminus (\upgen{#1}{#3}) = #2 \setminus (\upgen{#2}{#3})}
\newcommand{\simdown}[3]{#1 \setminus (\downgen{#1}{#3}) = #2 \setminus (\downgen{#2}{#3})}

\newcommand{\eqdown}[1]{\equiv^\downarrow_{#1}}
\newcommand{\upgen}[2]{{#2}^\uparrow_{#1}}
\newcommand{\downgen}[2]{{#2}^\downarrow_{#1}}
\newcommand{\dfrl}{\mathsf{d}}
\newcommand{\tset}[1]{[\![ #1 ]\!] }

\newcommand{\languagesym}{\mathtt{MLC}}
\newcommand{\logicsym}{\mathbf{MLC}}
\newcommand{\alg}{\mathfrak{A}}
\newcommand{\algb}{\mathfrak{B}}
\newcommand{\mapof}[2]{[\![ #1 (-,\vec{#2})]\!]}
\newcommand{\mapn}[2]{[\![ #1 (-,{#2}_1,...,{#2}_n]\!]}
\newcommand{\mapp}[1]{(\! (\chi) \!))}
\newcommand{\mapoff}{[\![ \fdia (-)]\!]}
\newcommand{\cO}{\mathcal{O}}

\def\labelitemi{--}
\section{Introduction}

Many modal logics, particularly those appearing in computer science and formal verification, are obtained by extending a basic modal language with operators given by least or greatest fixpoint definitions, relying on the Knaster-Tarski theorem \cite{knaster1928theoreme,tarski1955lattice}. A classic example is the propositional dynamic logic of regular programs \cite{fisc:prop79}, along with its extensions, like  concurrent $\mathtt{PDL}$ \cite{hare:dyna00}, or dynamic game logic \cite{pari:logi85}. In these logics, the \emph{Kleene star} is introduced as an operation on programs/games, where the program $\pi^*$ is understood as ``finite iteration'' of the program $\pi$. This is a least fixpoint construction.
Other examples of such \emph{modal fixpoint logics} include computation tree logic $\mathtt{CTL}$ \cite{clar:desi81} and its extension $\mathtt{CTL}^*$  \cite{emer:some86},  the logic of common knowledge \cite{halp:know90}, and of course the full \emph{modal $\mu$-calculus} \cite{koze:resu83} where a variable binding operation $\mu x.(-)$ is introduced to allow explicit least fixpoint definitions in the object language. 

A key theme in modal logic research is to develop tools for proving completeness for modal logics with respect to some state based or ``possible worlds'' semantics, usually Kripke semantics  but also topological semantics \cite{vanb:moda07}, neighborhood semantics \cite{chel:moda80} or coalgebraic semantics \cite{moss:coal99}. The standard modal $\mu$-calculus is interpreted on Kripke models and has an elegant complete axiomatization, introduced by Kozen in \cite{koze:resu83} and proved complete by Walukiewicz in \cite{walu:comp93}. It extends the standard axiom system for the minimal normal modal logic $\mathbf{K}$ with a single axiom and a single rule, jointly expressing the definition of $\mu x. \varphi(x)$ as the smallest pre-fixpoint for $\varphi(x)$. These are the \emph{pre-fixpoint axiom}:
$$ \varphi(\mu x. \varphi(x)) \rightarrow \mu x. \varphi(x) $$
and the \emph{induction rule}:
$$\frac{\varphi(\gamma) \rightarrow \gamma}{\mu x. \varphi(x) \rightarrow \gamma}$$
As opposed to the case of ``basic'' modal logics, there does not exist an extensive arsenal of general tools for proving completeness of modal fixpoint logics, of generality comparable to that of for example the Sahlqvist completeness theorem \cite{sahl:comp75}. Moreover, completeness proofs for concrete modal fixpoint logics can often be highly involved, witnessed especially by Reynold's completeness proof for $\mathtt{CTL}^*$ \cite{reyn:axio01} or Walukiewicz's completeness proof for the modal $\mu$-calculus. Some work in a more general direction has started to emerge, however: for example, a general completeness result for Venema's coalgebraic fixpoint logic introduced in \cite{vene:auto06} was recently proved in \cite{enqv:comp16}. In another direction,  Sahlqvist completeness with respect to a form of \emph{descriptive Kripke frames}\footnote{These are Kripke frames in which valuations are restricted to certain ``admissible subsets'', which are clopen sets in a Stone topology on the frame. Descriptive frames play an important role in completeness theory for modal logic \cite{blac:moda01}.} for modal $\mu$-calculi was established in \cite{bezh:sahl12}. But most relevant for this paper is the work of Santocanale and Venema, who prove a general completeness result, with respect to the standard Kripke semantics, for a class of modal fixpoint logics called \emph{flat modal fixpoint logics} \cite{sant:comp10}. These are extensions of basic modal logics with ``fixpoint connectives'' defined using only a single fixpoint variable, i.e. they are all fragments of the single-variable and hence of the \emph{alternation-free} modal $\mu$-calculus \cite{brad:moda98}. Schr\"{o}der and Venema later generalized the result to a coalgebraic setting in \cite{schr:flat10}.

The model construction being used in \cite{sant:comp10} can essentially be seen as a sort of argument by ``selection of a submodel of the canonical model''. This approach has a certain familiarity to it, as transforming the canonical model is a common strategy for completeness proofs in modal logic \cite{blac:moda01}. Furthermore, the kind of combinatorial problems that are involved in the completeness proof for the full $\mu$-calculus do not arise, since they are mostly due to the subtleties of alternating least- and greatest fixpoints.

In the same spirit,  we shall here investigate flat fragments of the \emph{two-way modal $\mu$-calculus}, in which the modal ``base logic'' is extended to include a converse modality. The difference with the standard $\mu$-calculus is that we allow modal operators to quantify both over the successors ($\fdia$) and the predecessors ($\bdia$) of the point of evaluation with respect to the accessibility relation. We can also see the logic as an \emph{axiomatic} extension of the minimal normal modal logic in a signature with two modalities $\fdia,\bdia$, obtained by adding the extra axioms:
$$p \rightarrow \fbox \bdia p \quad \quad \quad p \rightarrow \bbox \fdia p$$
The frame class determined by these axioms is the class of all frames with two accessibility relations, in which each accessibility relation is the converse of the other.

The two-way $\mu$-calculus is a well established logic. An automata-theoretic decision procedure for the two-way $\mu$-calculus was presented by Vardi in  \cite{vard:reas98}, and since the logic lacks the finite model property the satisfiability problem over finite models has been studied separately by Boja\'{n}czyk in \cite{boja:twow02}. The two-way $\mu$-calculus along with other \emph{enriched} modal mu-calculi \cite{satt:hybr01, bona:enri08}  has later gained attention in connection with description logics.  The two-way $\mu$-calculus is also closely related to \emph{guarded fixpoint logic}, see \cite{grad:guar02}. As far as we know, no completeness theorem for the two-way $\mu$-calculus is known so far.

The converse modality has also been studied in connection with other modal fixpoint logics and fragments of the two-way $\mu$-calculus: for example, the converse program constructor is  well known  in $\mathtt{PDL}$ \cite{stre:loop81}, and temporal logics  like $\mathtt{LTL}$ and $\mathtt{CTL}$ with converse have been studied for example in \cite{gore:ando14}. A version of the full computation tree logic $\mathtt{CTL}^*$ with a past modality has also been considered, and Reynolds extended his completeness theorem for $\mathtt{CTL}^*$ to include past operators in \cite{reyn:axio05}. For a decision procedure for the alternation-free fragment of the two-way $\mu$-calculus, see \cite{tana:deci05}.

The main result of this paper is a completeness theorem for Kozen's axioms, for a class of flat fixpoint logics with the converse modality.  As in \cite{sant:comp10}, we will need some constraints on the fixpoint connectives we consider\footnote{Santocanale and Venema also give a more general completeness result (using a more involved axiom system than the ones we consider here), by ``simulating'' arbitrary flat fixpoint logics using untied ``systems of fixpoint equations'', much like alternating tree automata can be simulated by non-deterministic ones. We will not follow this route here, but will briefly consider the issue in the concluding section.}, restricting attention to what we call \emph{disjunctive} fixpoint connectives (called ``untied formulas'' in \cite{sant:comp10}). We prove a generic completeness theorem for such logics, covering $\mathtt{CTL}$ with both a future and a past modality, and non-deterministic ``until'' and  ``since'' operators, as an immediate application. The proof is  necessarily different from that of Santocanale and Venema,  due to the fact that one of their main algebraic results for flat modal fixpoint logics - that disjunctive formulas correspond to \emph{finitary $\cO$-adjoints} \cite{sant:comp08} in the Lindenbaum-Tarski algebra - becomes simply false in the presence of the converse modality. This is explained in more detail in Section \ref{algebras}. 

The paper assumes some basic familarity with modal logic, algebras and fixpoint theory.

\section{Flat modal fixpoint logics with converse}

\subsection{Syntax and semantics}

We begin by introducing the syntax of basic modal logic extended with the backwards modality. Given a set of variables $\Prop$, we denote by $\languagesym(\Prop)$ the set of formulas defined by the following grammar:
$$\varphi := \bot  \mid p \mid \neg \varphi \mid \varphi \vee \varphi \mid \fdia \varphi \mid \bdia \varphi$$
where $p \in \Prop$. We define $\top = \neg \bot$, $\fbox \varphi = \neg \fdia \neg \varphi$, $\bbox \varphi = \neg \bdia \neg\varphi$, $\varphi \wedge \psi = \neg( \neg \varphi \vee \neg \psi )$, $\varphi \rightarrow \psi = \neg \varphi \vee \psi$ and $\varphi \leftrightarrow \psi = (\varphi \rightarrow \psi) \wedge (\psi \rightarrow \varphi)$. A formula $\varphi$ is said to be \emph{positive} in the variable $p$ if every occurrence of $p$ in $\varphi$ is in the scope of an even number of negations.

Throughout the paper we fix a variable $x$, which we call the \emph{recursion variable}, and a countable supply $\mathtt{Q}$ of variables $q_0,q_1,q_2,...$ called \emph{parameter variables}. 
\begin{defi}
An \emph{$n$-place fixpoint connective} is a formula $\chi(x,q_1,...,q_n) \in \languagesym(\{x\} \cup \mathtt{Q})$ that is positive in the variable $x$. A fixpoint connective $\chi(x,q_1,...,q_n)$ is said to be \emph{guarded} if every occurrence of the variable $x$ in $\chi(x,q_1,...,q_n)$ is in the scope of some modal operator. 
\end{defi}
Note that an $n$-place fixpoint connective is actually a formula in $n + 1$ free variables, but the  recursion variable $x$ should be viewed as an auxiliary device used to introduce an $n$-place connective by a least fixpoint definition. 

We will often write tuples of parameter variables as $\vec{q} = q_1,...,q_n$. Given a set $\Gamma$ of fixpoint connectives, we define the set of formulas of the extended language $\languagesym_\Gamma(\Prop)$ by the following grammar:
$$\varphi := \bot \mid p \mid \neg \varphi \mid \varphi \vee \varphi \mid \fdia \varphi \mid \bdia \varphi  \mid \sharp_\chi(\theta_1,...,\theta_n)$$
where $\chi(x,q_1,...,q_n) \in \Gamma$, and $\theta_1,...,\theta_n \in \languagesym_\Gamma(\Prop)$. The \emph{fixpoint nesting depth} of a formula $\varphi$ is the maximal number of nested occurrences of operators $\sharp_{\chi}$ in $\varphi$.

A central role will be played by the so called \emph{cover modality}, or nabla modality, which produces a modal formula from a \emph{set} of formulas  \cite{moss:coal99,walu:comp01,kupk:comp08}. We shall not formally introduce the nabla modality into the syntax as a primitive operator here, but merely use the symbol as an abbreviation:
$$\fnabla \{\varphi_1,...,\varphi_n\} := \fdia \varphi_1 \wedge ... \wedge \fdia \varphi_n \wedge \fbox (\varphi_1 \vee ... \vee \varphi_n)$$
$$\bnabla \{\varphi_1,...,\varphi_n\} := \bdia \varphi_1 \wedge ... \wedge \bdia \varphi_n \wedge \bbox (\varphi_1 \vee ... \vee \varphi_n)$$
Note that if we take the empty disjunction to be $\bot$ and the empty conjunction to be $\top$, we get $\fnabla \emptyset \Leftrightarrow \fbox \bot$ and $\bnabla \emptyset \Leftrightarrow \bbox \bot$.
It is well known that the box and diamond modalities can also be defined in terms of nablas, using the following standard equivalences:
$$\fdia \varphi \Leftrightarrow \fnabla \{\varphi, \top\} \quad \quad \fbox \varphi \Leftrightarrow \fnabla \emptyset \vee \fnabla \{\varphi\}$$
and similarly for $\bdia,\bbox$. 

The nabla modality already has an established place in modal logic: it paved the way to the \emph{coalgebraic} approach to modal logic, and was also independently used by Walukiewicz in his famous completeness proof for the modal $\mu$-calculus. In the context of flat modal fixpoint logics, it played the role of defining the \emph{untied formulas} for which Santocanale and Venema's completeness theorem is stated.

\begin{defi}
The set of \emph{forward-looking disjunctive formulas} in  $\languagesym(\{x\} \cup \mathtt{Q}) $ is defined by the following grammar:

$$\varphi := \bot \mid \top \mid x \mid \theta \wedge \varphi \mid \varphi \vee \varphi \mid \fnabla \Gamma $$
where $x$ does not occur in $\theta$ and $\Gamma$ is  a finite set of forward-looking disjunctive formulas. 

The set of \emph{backward-looking disjunctive formulas}  in  $\languagesym(\{x\} \cup \mathtt{Q}) $ is defined by the following grammar:

$$\varphi := \bot \mid \top \mid x \mid \theta \wedge \varphi \mid \varphi \vee \varphi \mid \bnabla \Gamma $$
where $x$ does not occur in $\theta$ and $\Gamma$ is  a finite set of backward-looking disjunctive formulas. 

A fixpoint connective  is said to be \emph{disjunctive} if it is either forward-looking disjunctive or backward-looking disjunctive. 
\end{defi}

What we call ``disjunctive'' here is called ``untied in $x$'' in \cite{sant:comp10}. Intuitively, disjunctive formulas avoid conjunctions between distinct subformulas containing the recursion variable $x$, with the exception of ``harmless'' conjunctions that are implicit in the cover modalities.

Kripke semantics for logics $\languagesym_\Gamma$ are defined as usual:
\begin{defi}
A Kripke frame is a pair $\mathfrak{F} = (W,R)$ where $W$ is a non-empty set and $R \subseteq W \times W$. Elements of $W$ will be referred to as states. Given a set of  propositional variables $\Prop$, a Kripke model $\bbS$ is defined to be a triple $(W,R,V)$ where $(W,R)$ is a Kripke frame and $V: \Prop \to \funP W$ is a valuation  of the variables. The \emph{truth set} of a formula in $\bbS = (W,R,V)$ is defined by the following recursion:
\begin{itemize}
\item $\tset{\bot}_\bbS = \emptyset$, $\tset{p}_\bbS = V(p)$,
\item $\tset{\neg \varphi}_\bbS = W \setminus \tset{\varphi}_\bbS$, $\tset{\varphi \vee \psi}_\bbS = \tset{\varphi}_\bbS \cup \tset{\psi}_\bbS$,
\item $\tset{\fdia \varphi}_\bbS = R^{-1}(\tset{\varphi}_\bbS)$ and $\tset{\bdia \varphi}_\bbS = R(\tset{\varphi}_\bbS)$,
\item $\tset{\sharp_{\chi}(\theta_1,...,\theta_n)}_\bbS = \bigcap \{Z \subseteq W \mid \tset{\chi(x,\theta_1,...,\theta_n)}_{\bbS[x \mapsto Z]} \subseteq Z\}$ for $\chi(x,q_1,...,q_n) \in \Gamma$.
\end{itemize}
Here, $R(Z) = \{w \in W \mid \exists v \in Z: \; v R w\}$, $R^{-1}(Z) = \{w \in W \mid \exists v \in Z: \; w R v\}$ and  $\bbS[x \mapsto Z]$ is the model $(W,R,V[x \mapsto Z])$ where the valuation $V[x \mapsto Z]$ is like $V$ except for mapping $x$ to $Z$. 

For $w \in W$ we write $\bbS,w \Vdash \varphi$ for $w \in \tset{\varphi}_\bbS$ and refer to $(\bbS,w)$ as a pointed Kripke model. We say that $\varphi$ is \emph{valid} in the Kripke semantics if $\bbS,w \Vdash \varphi$ for every pointed Kripke model $(\bbS,w)$, and we write $\Vdash \varphi$ to express this.
 \end{defi}

\begin{exam}
Consider a one-place fixpoint connective $\chi_1(x,q) = \fbox x \vee q$ and a one-place connective $\chi_2(x) = \bbox x \vee q$. The formula $\sharp_{\chi_1} \varphi$ then says that one cannot move forward along the accessibility relation forever, without eventually visiting a point where $\varphi$ is true. Similarly, the formula $\sharp_{\chi_2} \varphi$ says that that one cannot move backward along the accessibility relation forever, without eventually visiting a point where $\varphi$ is true. 

In particular, consider the formula:
$$\neg \sharp_{\chi_1}(\neg\sharp_{\chi_2}\bot)$$
This formula says that there is some infinite ``forward'' path, such that from each point on this path there are no infinite ``backward'' paths. This formula is satisfiable, but  has no finite models, so just like the full two-way $\mu$-calculus the flat fixpoint logics we consider here do not generally have the finite model property. 
\end{exam}

\begin{defi}
Let $X,Y$ be any sets and $R \subseteq X \times Y$ a binary relation. Then we say that $R$ is \emph{full} (with respect to $X,Y$) if $Y \subseteq R(X)$ and $X \subseteq R^{-1}(Y)$.
\end{defi}
The following is a standard fact:

\begin{prop}
Let $(\bbS,w)$ be a pointed Kripke model. Then $\bbS,w\Vdash \fnabla \Psi$ iff there is a full relation $Z \subseteq R(w) \times \Psi$ such that $\bbS,v \Vdash \psi$ whenever $v Z \psi$. Similarly, $\bbS,w\Vdash \bnabla \Psi$ iff there is a full relation $Z \subseteq R^{-1}(w) \times \Psi$ such that $\bbS,v \Vdash \psi$ whenever $v Z \psi$
\end{prop}

Given a fixpoint connective $\chi(x,q_1,...,q_n)$ and a model $\bbS$ we shall denote the monotone map sending a subset $Z$ of $W$ to  $\tset{\chi(x,\theta_1,...,\theta_n)}_{\bbS[x \mapsto Z]}$ by $\mapn{\chi}{\theta}_\bbS$. The truth condition for $\sharp_\chi$ then expresses that $\tset{\sharp_{\chi}(\theta_1,...,\theta_n)}_\bbS$ is the least pre-fixpoint for the map $\mapn{\chi}{\theta}_\bbS$, which by the Knaster-Tarski theorem is the least fixpoint of the same map. Another direct consequence of the Knaster-Tarski theorem is the following:
\begin{prop}
\label{p:approximants}
Let $\chi(x,q_1,...,q_n)$ be any fixpoint connective in $\Gamma$ and let $\theta_1,...,\theta_n$ be a tuple of formulas in $\languagesym_\Gamma(\Prop)$.  Then we have: $$\Vdash \chi^k(\bot,\theta_1,...,\theta_n) \rightarrow \sharp_\chi (\theta_1,...,\theta_n)$$ for all $k \in \omega$, where  $\chi^k(\bot,\theta_1,...,\theta_n)$ is defined by the recursion:
\begin{itemize}
\item $\chi^0(\bot,\theta_1,...,\theta_n) = \chi(\bot,\theta_1,...,\theta_n)$,
\item $\chi^{m + 1}(\bot,\theta_1,...,\theta_n) = \chi(\chi^m(\bot,\theta_1,...,\theta_n),\theta_1,...,\theta_n)$.
\end{itemize}
\end{prop}

\subsection{Axiom systems}

In this section we introduce Hilbert-style proof systems for flat fixpoint logics. Let $\Gamma$ be any set of fixpoint connectives. We present the proof system $\logicsym_\Gamma$ by the following axioms and rules. As axioms, we take all propositional tautologies and the following axiom schemata:
\begin{itemize}
\item $\neg \fdia \bot $ and $\neg \bdia \bot$, 
\item $\fdia(\varphi \vee \psi) \leftrightarrow (\fdia \varphi \vee \fdia \psi)$ and $\bdia(\varphi \vee \psi) \leftrightarrow (\bdia \varphi \vee \bdia \psi)$,
\item $\varphi \rightarrow \fbox \bdia \varphi$ and $\varphi \rightarrow \bbox \fdia \varphi$,
\item $\chi(\sharp_\chi (\theta_1,...,\theta_n),\theta_1,...,\theta_n) \rightarrow \sharp_\chi (\theta_1,...,\theta_n)$, for $\chi(x,q_1,...,q_n) \in \Gamma$.
\end{itemize}

As a rule schema we take Replacement of Equivalents:
$$\frac{\varphi \leftrightarrow \psi}{\theta[\varphi/p] \leftrightarrow \theta[\psi/p]}$$
where $\theta[\varphi/p]$  and $\theta[\psi/p]$ denote the result of uniformly replacing $\varphi$ and $\psi$ for $p$ in $\theta$, respectively.  Finally, we take the Kozen-Park Induction Rule:
$$\frac{\chi(\gamma,\theta_1,...,\theta_n) \rightarrow \gamma}{\sharp_\chi (\theta_1,...,\theta_n) \rightarrow \gamma}$$
We write $\logicsym_\Gamma \vdash \varphi$ to say that $\varphi$ is derivable in the system $\logicsym_\Gamma$, or just $\vdash \varphi$ if the system $\logicsym_\Gamma$ is clear from context. We may also write $\varphi \vdash \psi$ instead of $\vdash \varphi \rightarrow \psi$. The proof of the following proposition is standard:
\begin{prop}[Soundness]
For every formula $\varphi \in \languagesym_\Gamma(\Prop)$, if $\logicsym_\Gamma \vdash \varphi$ then $\Vdash \varphi$.
\end{prop}
We write $\logicsym$ for $\logicsym_\emptyset$.

\begin{prop}
We have: $$\logicsym \vdash \fbox(\varphi \rightarrow \psi) \rightarrow (\fbox \varphi \rightarrow \fbox \psi)$$ and 
$$\logicsym \vdash \bbox(\varphi \rightarrow \psi) \rightarrow (\bbox \varphi \rightarrow \bbox \psi).$$ Also, the standard Necessitation Rule is derivable in $\logicsym$ for both box modalities $\fbox$ and $\bbox$.
\end{prop}
\begin{proof}
For $\logicsym \vdash \fbox(\varphi \rightarrow \psi) \rightarrow (\fbox \varphi \rightarrow \fbox \psi)$, just rewrite this formula equivalently as:
$$\fdia(\varphi \wedge \neg \psi) \vee \fdia \neg \varphi \vee  \neg \fdia \neg \psi$$
Applying the additivity axiom for $\fdia$ we get the equivalent formula:
$$\fdia((\varphi \wedge \neg \psi) \vee \neg \varphi) \vee  \neg \fdia \neg \psi$$
which reduces to:
$$\fdia(\neg \psi \vee \neg \varphi) \vee  \neg \fdia \neg \psi$$
Applying additivity again  we get the equivalent formula:
$$\fdia \neg \psi \vee \fdia \neg \varphi \vee \neg \fdia \neg \psi$$
which is obviously provable. The theorem $\logicsym \vdash \bbox(\varphi \rightarrow \psi) \rightarrow (\bbox \varphi \rightarrow \bbox \psi)$ is proved in the same manner, and the necessitation rule is derived using replacement of equivalents applied to the provable formulas $\fbox \top$ and $\bbox \top$.
\end{proof}

We say that the system $\logicsym_\Gamma$ is \emph{Kripke complete} if for every formula $\varphi \in \languagesym_\Gamma(\Prop)$, if $\Vdash \varphi$ then $\logicsym_\Gamma \vdash \varphi$. We can now state the main theorem of the paper:

\begin{theo}
\label{t:main}
Let $\Delta$ be a set of disjunctive and guarded fixpoint connectives. Then the system $\logicsym_\Delta$ is Kripke complete.
\end{theo}

This result can be strengthened to drop the guardedness constraint, as we will see in Section \ref{guardedness}.
Before we proceed with the proof of Theorem \ref{t:main}, we briefly consider the canonical model construction for logics $\logicsym_\Gamma$:
\begin{defi}
Let $\Gamma$ be a set of fixpoint connectives. The \emph{canonical model} for the logic $\logicsym_\Gamma$, denoted $\mathfrak{C}_\Gamma$, is the triple $(W^C,R^C,V^C)$ where:
\begin{itemize}
\item $W^C$ is the set of maximal consistent sets of formulas in $\languagesym_\Gamma(\Prop)$,
\item $\Theta_1 R^C \Theta_2$ iff for all formulas $\varphi$ such that $\fbox \varphi \in \Theta_1$, we have $\varphi \in \Theta_2$,
\item $V^C(p) = \{\Theta \in W^C \mid p \in \Theta\}$.
\end{itemize}
\end{defi}
Technically, the canonical model can in fact be described as a descriptive model for the fixpoint logic $\logicsym_\Gamma$ in the sense of \cite{bezh:sahl12}, but this theme will not be further explored here. 
\begin{prop}
\label{p:canon}
Let   $\mathfrak{C}_\Gamma = (W^C,R^C,V^C)$ be the canonical model for the logic $\logicsym_\Gamma$. Then the following are equivalent:
\begin{enumerate}
\item $\Theta_1 R^C \Theta_2$,
\item for all $\varphi \in \languagesym_\Gamma(\Prop)$: $\bbox \varphi \in \Theta_2$ implies $\varphi \in \Theta_1$,
\item for all $\varphi \in \languagesym_\Gamma(\Prop)$: $\varphi \in \Theta_2$ implies $\fdia \varphi \in \Theta_1$,
\item for all $\varphi \in \languagesym_\Gamma(\Prop)$: $\varphi \in \Theta_1$ implies $\bdia \varphi \in \Theta_2$.
\end{enumerate}
\end{prop}

\begin{proof}
For $(1) \Leftrightarrow (2)$: suppose $\Theta_1 R^C \Theta_2$, and suppose that  $\bbox \varphi \in \Theta_2$. If $\varphi \notin \Theta_1$ then $\neg \varphi \in \Theta_1$, so $\fbox \bdia \neg \varphi \in \Theta_1$, i.e. $\fbox \neg \bbox \varphi \in \Theta_1$. By definition of $R^C$ we get $\neg \bbox \varphi \in \Theta_2$, which means $\Theta_2$ is inconsistent, contradiction.  The converse implication is proved similarly. 

For $(1) \Leftrightarrow (3)$, suppose $\Theta_1 R^C \Theta_2$, and suppose $\varphi \in \Theta_2$. If $\fdia \varphi \notin \Theta_1 $ then $\neg \fdia \varphi \in \Theta_1$, so $\fbox \neg \varphi \in \Theta_1$. By definition of $R^C$ we get $\neg \varphi \in \Theta_2$, so $\Theta_2$ is inconsistent, contradiction. The converse implication is proved by a dual argument. Finally, the equivalence $(2) \Leftrightarrow (4)$ is proved in the same manner. 
\end{proof}

\begin{prop}[Existence Lemma]
\label{p:existence}
Let $\Theta$ be an element of the canonical model $\mathfrak{C}_\Gamma$, and let $\varphi \in \languagesym_\Gamma(\Prop)$. If $\fdia \varphi \in \Theta$, then $\Theta$ has an $R^C$-successor $\Theta'$ with $\varphi \in \Theta'$. Similarly, if $\bdia \varphi \in \Theta$, then $\Theta$ has an $R^C$-predecessor $\Theta'$ with $\varphi \in \Theta'$.
\end{prop}
\begin{proof}
Standard argument using Lindenbaum's lemma, and using Proposition \ref{p:canon} for the case involving the converse modality. 
\end{proof}

\subsection{Algebraic semantics and $\cO$-adjoints}
\label{algebras}
The heart of Santocanale and Venema's completeness  proof for flat modal fixpoint logics is algebraic. One of their key observations is that operators on free algebras for flat modal fixpoint logics definable by disjunctive formulas   have the nice property of being \emph{finitary $\cO$-adjoints}. As a corollary, it follows that free algebras for flat modal fixpoint logics are \emph{constructive}, i.e. the existing least fixpoints in such algebras arise as infinitary joins of their finite approximants.

Here, algebraic models will play only a minor role in the completeness proof, where they will be useful to eliminate the guardedness constraint in Theorem \ref{t:main} in Section \ref{guardedness}. The main reason for considering the algebraic semantics for flat modal fixpoint logics here is rather to point out a \emph{negative} result: as we shall see, disjunctive formulas do \emph{not} correspond to finitary $\mathcal{O}$-adjoints in free algebras for flat fixpoint logics with converse, so the proof strategy used by Santocanale and Venema to prove constructiveness of free algebras will not work here. (Constructiveness of free algebras for the logics $\logicsym_\Gamma$ does however follow as a corollary to our completeness proof.)

\begin{defi}
A \emph{modal algebra} is a tuple $\alg = (A,1_\mathfrak{A},0_\mathfrak{A},\vee_\mathfrak{A},\neg_\mathfrak{A},\fdia_\mathfrak{A},\bdia_\mathfrak{A})$ such that $(A,1_\mathfrak{A},0_\mathfrak{A},\vee_\mathfrak{A},\neg_\mathfrak{A})$ is a boolean algebra and $\fdia_\mathfrak{A},\bdia_\mathfrak{A} : A \to A$ are unary operations that preserve the bottom element $0_\mathfrak{A}$, and are additive. That is, the algebra satisfies the equations:
$$\fdia(x \vee y) = \fdia x \vee \fdia y \quad \quad \quad \bdia(x \vee y) = \bdia x \vee \bdia y $$
Such an algebra is called a \emph{residuated modal algebra} if the following equivalence holds for all $a,b \in A$:
$$\fdia_\alg a \leq b \Leftrightarrow a \leq \bbox_\alg b$$
and 
$$\bdia_\alg a \leq b \Leftrightarrow a \leq \fbox_\alg b$$
where as usual, $x \leq y$ is an abbreviation for  $x \wedge y = x$.
\end{defi}

\begin{prop}
A modal algebra  $\mathfrak{A} = (A,1_\mathfrak{A},0_\mathfrak{A},\vee_\mathfrak{A},\neg_\mathfrak{A},\fdia_\mathfrak{A},\bdia_\mathfrak{A})$ is residuated if, and only if, it validates the equations:
$$  x \leq \fbox \bdia x \quad \quad\quad x \leq \bbox \fdia x$$
\end{prop}
\begin{proof}
Suppose $\alg$ is residuated. Since we have $a \leq a$ for all $a \in A$, we have $\bdia_\alg a \leq \bdia_\alg a$ since  $\bdia_\alg$ is monotone, hence $a \leq \fbox_\alg \bdia_\alg a$ by residuatedness. The equation $x \leq \bbox_\alg \fdia_\alg x$ is proved in the same manner. 

Suppose the two equations hold in $\alg$. If $\fdia_\alg a \leq b$ then $ \bbox_\alg \fdia_\alg \leq \bbox_\alg b$, and since $a \leq \bbox_\alg \fdia_\alg a$ we get $a \leq \bbox_\alg b$ as required.  Conversely, if $a \leq \bbox_\alg b$ then $\bdia_\alg \neg_\alg b \leq \neg_\alg a $ so $\fbox_\alg \bdia_\alg \neg_\alg b \leq \fbox_\alg \neg_\alg a$, and since $\neg_\alg b \leq \fbox_\alg \bdia_\alg \neg_\alg b$ we get $\neg_\alg b \leq \fbox_\alg \neg_\alg a$, so $\fdia_\alg a \leq b$. The other equivalence is proved by a similar argument.  
\end{proof}

For any modal formula $\varphi(p_1,...,p_n)$ with $n$ free variables, and for any given modal algebra $\alg$, we can recursively define an $n$-place map $\varphi_\alg : A^n \rightarrow A$ corresponding to $\varphi$ in the usual manner. To be more precise, we define: $$\varphi_\alg(a_1,...,a_n) = h(\varphi(p_1,...,p_n))$$ where $h$ is the unique homomorphism from the term algebra of the language $\mathtt{MLC}_\emptyset(\{p_1,...,p_n\})$ to $\alg$ extending the valuation map $V : \Prop \to A$ that maps each variable $p_i$ to $a_i$. 

\begin{defi}
Let $\Gamma$ be a set of fixpoint connectives. A residuated modal algebra 
$\alg$
 is called a \emph{$\logicsym_\Gamma$-algebra} if, for every fixpoint connective $\chi(x,\vec{q}) \in \Gamma$ and every tuple $\vec{b} \in A^n$, the  map $\chi_\alg(-,\vec{b}) : A \to A$ has a least pre-fixpoint (with respect to the order $\leq$).
\end{defi}

If $\alg$ is a $\logicsym_\Gamma$-algebra, then we can define the map $\varphi_\alg : A^n \rightarrow A$ for any formula $\varphi(p_1,...,p_n)$ in $\languagesym_\gamma(\Prop)$ with $n$ free variables, extending the recursion by setting $(\sharp_\chi)_\alg(a_1,...,a_n)$ to be the least pre-fixpoint of the map $\chi_\alg(-,a_1,...,a_n) : A \to A$. We say that the formula  $\varphi(p_1,...,p_n)$ is \emph{valid} on $\alg$ and write $\alg \vDash \varphi(p_1,...,p_n)$ if $\varphi_\alg(a_1,...,a_n) = 1_\alg$ for all $a_1,...,a_n \in A$.

The \emph{free $\logicsym_\Gamma$-algebra} $\mathfrak{A} = (A,1_\mathfrak{A},0_\mathfrak{A},\vee_\mathfrak{A},\neg_\mathfrak{A},\fdia_\mathfrak{A},\bdia_\mathfrak{A})$ is defined by the standard Lindenbaum-Tarski construction, as the quotient of the term algebra of the language $\languagesym_\Gamma(\Prop)$ by the equivalence relation  $\equiv$  defined by: $$\varphi \equiv \psi \text{ iff } \logicsym_\Gamma \vdash \varphi \leftrightarrow \psi$$ It is standard to prove that this is a $\logicsym_\Gamma$-algebra, and that a formula $\varphi$ is provable in $\logicsym_\Gamma$ if and only if it is valid on the free $\logicsym_\Gamma$-algebra. So we get:

\begin{prop}
\label{p:algcompleteness}
The  logic $\logicsym_\Gamma$ is sound and complete with respect to $\logicsym_\Gamma$-algebras, i.e. $\logicsym_\Gamma \vdash \varphi$ iff $\alg \vDash \varphi$ for every $\logicsym_\Gamma$-algebra $\alg$.
\end{prop}
We now revisit  the topic of finitary $\cO$-adjoints:

\begin{defi}
Let $\alg$ be a modal algebra and $f : A^n \to A$ any map. Then $f$ is called an \emph{$\cO$-adjoint} if there is a map $\cO_f : A \to \funP A^n$ mapping each $\vec{a} \in A^n$ to a finite set $\cO_f(\vec{a}) \subseteq A^n$, and such that for all $\vec{a} \in A^n, b \in A$: $f(\vec{a}) \leq b$ iff $a_i \leq d_i$ for some $(d_1,...,d_n) \in \cO_f(b)$ and all $i \in \{1,...,n\}$. 

A set $B \subseteq A$ is said to be $\cO_f$-closed if $b \in B$ implies $d_i \in B$ for all $(d_1,...,d_n) \in  \cO_f(b)$ and $i \in \{1,...,n\}$. The map $f$ is called a \emph{finitary $\cO$-adjoint} if for each $b \in A$, there is a finite $\cO_f$-closed set containing $b$. 
\end{defi}

A crucial result in \cite{sant:comp10} says (in rough terms) that: the maps corresponding to disjunctive formulas in the free algebras for flat modal fixpoint logics are finitary $\cO$-adjoints. Our main observation in this section is that this is not true for free $\logicsym_\Gamma$-algebras, even for very simple disjunctive fixpoint connectives like the one-place connective $y \vee \fdia x $. Note the meaning of this fixpoint connective: the formula $\sharp_{y \vee \fdia x } p$ expresses finite reachability of some point that satisfies $p$.

\begin{prop}
Let $\Gamma$ be any set of fixpoint connectives containing $y \vee \fdia x$ where  $y \in \mathtt{Q}$  is  any parameter variable. Then the map $(y \vee \fdia x)_\alg : A \times A \to A$ is not a finitary $\cO$-adjoint in the free $\logicsym_\Gamma$-algebra $\alg$.
\end{prop}

\begin{proof}
We shall  drop the subscript referring to $\alg$ when denoting operations in $\alg$, for the sake of readability. Let $f$ denote the map $(a,b) \mapsto  a \vee \fdia b$. Suppose  that $\cO_f : A \to \funP A$ were a map witnessing finitary $O$-adjointness of  the map $f : A \times A \to A$. We first claim that for any $\cO_f$-closed set $C$, and any $c \in C$, we have $\bbox c \in C$ also. To see this, let $c \in C$.  Since $c \vee \fdia ( \bbox c)  \leq c$, there is some $(a',b') \in \cO_f(c)$ such that $c \leq a'$ and $\bbox c \leq b'$. But since $a' \leq a'$  and $b' \leq b'$, and  $(a',b') \in \cO_f(c)$, it follows that $a' \vee \fdia b' \leq c$, hence $\fdia b' \leq  c$ so $b' \leq \bbox c$. So in fact $b' = \bbox c$, hence $\bbox c \in C$.  

Now apply this claim to $c = 0$, the bottom element of $\alg$.  Let $C$ be the smallest $\cO_f$-closed set containing $0$. We get $\bbox 0\in C$, $\bbox \bbox 0 \in C$, $\bbox \bbox \bbox 0 \in C$... and so on. Since it is easy to construct Kripke models in which  the formulas   $\bbox\bot$, $\bbox \bbox \bot$, $\bbox \bbox \bbox \bot$ all have pairwise different truth sets, by soundness of $\logicsym_\Gamma$  the equivalence classes of these formulas are distinct elements of the algebra $\alg$. In other words, the elements $\bbox 0$, $\bbox \bbox 0 $, $\bbox \bbox \bbox 0 $... are all distinct elements of $C$,  which is a contradiction with our assumption that $\cO_f$ was finitary.
\end{proof}

What this means in practice is that we have to take a different route to the completeness result than the one taken by Santocanale and Venema: rather than establishing constructiveness of the free $\logicsym_\Gamma$-algebra in a first step, we will directly construct a model for a consistent formula, using the Kozen-Park induction rule ``on the fly'' to ensure satisfaction of the least fixpoint formulas one by one. This is essentially the idea of ``foci'', which has been used to produce SAT-algorithms for temporal logics \cite{lang:focu01} and more recently for the alternation-free $\mu$-calculus \cite{haus:glob16}. For the formal construction, we will use \emph{networks}, a well known technique in modal logic that is also used in \cite{sant:comp10}. See \cite{blac:moda01} for an introduction to the technique.

\section{Networks}
\subsection{Basic definitions}

In this section we introduce the basics of \emph{networks}, which we will use to construct models for consistent formulas. We start with a familiar construction from propositional dynamic logic and other modal fixpoint logics, that of \emph{Fischer-Ladner closure}. Essentially, Fischer-Ladner closed sets of formulas are closed under subformulas, single negations and single unfoldings of fixpoints. Besides this we shall add a couple of minor extra conditions here that will be convenient in some of the later proofs. 

\begin{defi}
Let $\Gamma$ be a set of fixpoint connectives, and let $\Sigma$ be a set of formulas in $\languagesym_\Gamma(\Prop)$. We say that $\Sigma$ is Fischer-Ladner closed if:
\begin{itemize}
\item if $\varphi \in \Sigma$ then every subformula of $\varphi$ is in $\Sigma$ as well,
\item if $\varphi \in \Sigma$ and is not of the form $\neg \psi$ for any $\psi$, then $\neg \varphi \in \Sigma$,
\item if $\chi(x,\vec{q}) \in \Gamma$ and $\sharp_\chi\vec{\theta} \in \Sigma$, then  $\chi(\sharp_\chi \vec{\theta},\vec{\theta}) \in \Sigma$ and $\chi(\bot,\vec{\theta}) \in \Sigma$,
\item $\fbox \bot \in \Sigma$ and $\bbox \bot \in \Sigma$.
\end{itemize}
\end{defi}

The proof of the following proposition is entirely standard and is therefore omitted:
\begin{prop}
\label{p:flclosure}
Let $\varphi$ be any formula in $\languagesym_\Gamma(\Prop)$. Then there exists a Fischer-Ladner closed and finite set $\Sigma$ with $\varphi \in \Sigma$.
\end{prop}

We call the smallest Fischer-Ladner closed set containing $\varphi$ the \emph{Fischer-Ladner closure} of $\varphi$. 

\begin{defi}
Let $\Sigma$ be a Fischer-Ladner closed subset of $\languagesym_\Gamma(\Prop)$. A set of formulas $\Theta$ is then called a \emph{$\Sigma$-atom} if $\Theta$ is equal to $\Psi \cap \Sigma$ for some maximal consistent subset $\Psi \subseteq \languagesym_\Gamma(\Prop)$.
\end{defi}
For a Fischer-Ladner closed set $\Sigma$ we let $ \mathsf{At}(\Sigma)$ denote the set of $\Sigma$-atoms.
It follows immediately from Lindenbaum's lemma that any consistent formula $\varphi \in \Sigma$ is a member of some $\Sigma$-atom. The proposition below lists some basic properties of atoms, which will all be familiar:
\begin{prop}
Let $\Theta$ be a $\Sigma$-atom. Then:
\begin{itemize}
\item $\varphi \vee \psi \in \Theta$ iff $\varphi \in \Theta$ or $\psi \in \Theta$, for $\varphi \vee \psi \in \Sigma$, 
\item $\varphi \wedge \psi \in \Theta$ iff $\varphi \in \Theta$ and $\psi \in \Theta$, for $\varphi \wedge \psi \in \Sigma$, 
\item $\varphi \in \Theta$ iff $\sim\! \varphi \notin \Theta$, for $\varphi \in \Sigma$, where we define $\sim\! \neg \psi = \psi$ and $\sim \!\psi = \neg\psi$ if the main connective in $\psi$ is not $\neg$,
\item $\sharp_\chi \vec{\theta} \in \Theta$ iff $\chi(\sharp_\chi \vec{\theta},\vec{\theta}) \in \Theta$, for $\chi(x,\vec{q}) \in \Gamma$ and $\sharp_\chi \vec{\theta} \in \Theta$.
\end{itemize}
\end{prop}

From now on we fix a set of fixpoint connectives $\Gamma$ and a finite, Fischer-Ladner closed $\Sigma \subseteq \languagesym_\Gamma(\Prop)$.
An important concept that we shall borrow from \cite{haus:glob16} is that of a \emph{deferral} of a fixpoint formula:
\begin{defi}
Let $\Sigma$ be a Fischer-Ladner closed subset of $\languagesym_\Gamma(\Prop)$ and let $\chi(x,\vec{q}) \in \Gamma$.
A \emph \emph{potential $\sharp_\chi \vec{\theta}$-deferral} is a formula of the form $\psi(x,\vec{q})$ which is a subformula of $\chi(x,\vec{q})$ in which the variable $x$ occurs at least once. A potential $\Sigma$-deferral is a potential $\sharp_\chi \vec{\theta}$-deferral for some $\sharp_\chi \vec{\theta} \in \Sigma $.
\end{defi}

We shall let $D(\Sigma)$ denote the set of potential $\Sigma$-deferrals. We let $\dfrl$ denote the size of the set $D(\Sigma)$ (which is of course finite since $\Sigma$ is), and we fix a bijective enumeration $\eta : \{1,...,\dfrl\} \to D(\Sigma)$.

We can now define the basic concept of a network. A network will be a certain directed acyclic graph, with atoms assigned to its nodes, representing an approximation of a model for some given formula. We will use the following notation: if $G$ is a directed graph, we write $\underline{G}$ for the set of nodes of $G$, and we write $u \edge{G} v$ for $u,v \in \underline{G}$ to say that there is an edge from $u$ to $v$ in $G$. In this case we say that $v$ is a successor of $u$, and that $u$ is a predecessor of $v$. The \emph{opposite} of a directed graph $G$, denoted $G^{\mathit{op}}$, is defined by setting $\underline{G}^{\mathit{op}} = \underline{G}$ and $u \edge{G^{\mathit{op}}} v$ iff $v \edge{G} u$.

\begin{defi}
Let $\Sigma$ be a finite, Fischer-Ladner closed set of formulas. A \emph{$\Sigma$-prenetwork} $\mathcal{N} = (G,L,S_F,S_P)$ consists of a DAG (directed acyclic graph) $G$ together with a labelling function $L : \underline{G} \to \mathsf{At}(\Sigma)$, and such that:
\begin{itemize}
\item If $\fbox \varphi \in L(u)$ and $u \edge{G} v$ then $\varphi \in L(v)$.
\item If $\bbox \varphi \in L(v)$ and $u \edge{G} v$ then $\varphi \in L(u)$.
\end{itemize}
and $S_F$, $S_P$ are subsets of $\underline{G}$. 
 A $\Sigma$-prenetwork $\mathcal{N} = (G,L,S_F,S_P)$ is called a \emph{$\Sigma$-network} if the following conditions hold:
\begin{itemize}
\item  If $u \in S_F$ then for each formula $\fdia \varphi \in L(u)$ and each $i \in \{1,...,\dfrl\}$ there exists a successor $v^\varphi_i$ of $u$ such that $\varphi \in L(v^\varphi_i)$, and such that  $L(v^\varphi_1) =  ... = L(v^\varphi_\dfrl)$ for each $\fdia \varphi \in L(u)$ and  $v^\varphi_i \neq v^\psi_j$ whenever $\varphi \neq \psi$ or $i \neq j$. 
\item If $u \in S_P$ then for each formula $\bdia \varphi \in L(u)$ and each $i \in \{1,...,\dfrl\}$ there exists a predecessor $v^\varphi_i$ of $u$ such that $\varphi \in L(v^\varphi_i)$, and such that  $L(v^\varphi_1) =  ... = L(v^\varphi_\dfrl)$ for each $\bdia \varphi \in L(u)$ and  $v^\varphi_i \neq v^\psi_j$ whenever $\varphi \neq \psi$ or $i \neq j$. 
\end{itemize}
A node $u$ is called a \emph{head node} if it has no successors, and a \emph{tail node} if it has no predecessors.
\end{defi}
The set $S_F$ is called the set of \emph{forward saturated nodes} in $\mathcal{N}$, and the set $S_P$ is called the set of \emph{backward saturated nodes} in $\mathcal{N}$. The idea is that if a node belongs to $S_F$, then it has all the successor nodes it will ever need, and we are not allowed to add any more successors later when we extend the network $\mathcal{N}$ in our step-by-step construction of a model. Similarly, a node in $S_P$ has all the predecessor nodes it will ever need, and we are not allowed to add any more predecessors later.

Given a directed graph $G$, we write $u \edge{G*} v$ if $u \edge{G} v$ or there is  some finite path: $$u \edge{G} w_1 \edge{G} ... \edge{G} w_k \edge{G} v$$ in $G$.
In other words,  $\edge{G*}$ is defined to be  the transitive closure of $\edge{G}$, and \emph{not} the reflexive transitive closure. So if $G$ is a DAG, the relation $\edge{G*}$ is always irreflexive.

\begin{defi}
A network $\mathcal{N} = (G,L)$ is said to be \emph{anticonfluent} if for any nodes $u,v,v',w \in \underline{G}$ with $v \neq v'$ it is not the case  that $u \gpath{G} v$, $u \gpath{G} v'$, $v \gpath{G} w$ and $v' \gpath{G} w$.
\end{defi}

\begin{defi}
Let $\mathcal{N} = (G,L,S_F,S_P)$ and $\mathcal{N}' = (G',L',S_F',S_P')$ be $\Sigma$-prenetworks. We say that $\mathcal{N}'$ \emph{contains} the prenetwork $\mathcal{N}$  and write $\mathcal{N} \subseteq \mathcal{N}'$ if:
\begin{itemize}
\item $G$ is a subgraph of $G'$ ($\underline{G} \subseteq \underline{G}'$ and $u \edge{G} v$ iff $u \edge{G'} v$ for $u,v \in \underline{G}$),
\item $L = L' {\upharpoonright}_{\underline{G}}$ ($L$ is the restriction of $L'$ to $\underline{G}$),
\item $S_F \subseteq S_F'$ and $S_P \subseteq S_P'$.
\end{itemize}
\end{defi}

The relation of ``being contained in'' is too weak for our purposes, since it does not reflect the intuition that we should not add new successors or new predecessors to nodes in $S_F$ or in $S_P$, respectively. For this reason we need the stronger notion of a \emph{subnetwork}:
\begin{defi}
A network $\mathcal{N} = (G,L,S_F,S_P)$ is said to be a \emph{subnetwork} of $\mathcal{N}' = (G',L',S_F',S_P')$, written $\mathcal{N} \sqsubseteq \mathcal{N}'$, if $\mathcal{N} \subseteq \mathcal{N}'$ and, for all $u \in \underline{G}$:
\begin{itemize}
\item If $u \in S_F$ and $u \edge{G'} v$ for $v\in \underline{G}'$, then $v \in \underline{G}$ also.
\item If $u \in S_P$ and $v \edge{G'} u$ for $v\in \underline{G}'$, then $v \in \underline{G}$ also.
\end{itemize}
\end{defi}
\subsection{Basic operations on networks}

We will need some rudimentary set theoretic operations on networks.
\begin{defi}
Let $\{\mathcal{N}_i\}_{i \in I}$ be a family of $\Sigma$-prenetworks, where each $\mathcal{N}_i = (G_i,L_i,S_F^i,S_P^i)$, and suppose that for all $i,j \in I$: if $u \in \underline{G}_i \cap \underline{G}_j$ then $L_i(u) = L_j(u)$. Then we define the union $\bigcup_{i \in I} \mathcal{N}_i$ of these prenetworks to be the tuple $(G',L',S_F',S_P')$, where:
\begin{itemize}
\item $\underline{G'} = \bigcup_{i \in I} \underline{G}_i$,
\item $u \edge{G'} v$ iff for some $i \in I$, $u,v \in \underline{G}_i$ and $u \edge{G_i} v$,
\item $L'(u) = L_i(u)$ for $u \in \underline{G}_i$,
\item $S_F' = \bigcup_{i \in I} S_F^i$ and $S_P' = \bigcup_{i \in I} S_P^i$.
\end{itemize}
\end{defi}
The following proposition is obvious. 
\begin{prop}
Suppose that $\{\mathcal{N}_i\}_{i \in I}$ is a family of $\Sigma$-networks. If the union $\bigcup_{i \in I} \mathcal{N}_i$ is defined, then it is also a $\Sigma$-network.
\end{prop}

\begin{defi}
Let $\mathcal{N} = (G,L,S_F,S_P)$ be a prenetwork and let $ X \subseteq \underline{G}$.  We define the restriction $\mathcal{N} {\upharpoonright} X$ to be the tuple $(G',L',S_F',S_P')$ where:
\begin{itemize}
\item $G'$ is the unique subgraph of $G$ with $\underline{G}' = X$,
\item $L' = L{\upharpoonright}_{X}$,
\item $S_F' = S_F \cap X$ and  $S_P' = S_P \cap X$.
\end{itemize}
We define the relative complement $\mathcal{N}\setminus X$ for $ X \subseteq \underline{G}$ to be the prenetwork $\mathcal{N}{\upharpoonright}(\underline{G} \setminus X)$.
\end{defi}
Note that the restriction $\mathcal{N}{\upharpoonright}_X$ need \emph{not} be a network even if $\mathcal{N}$ is, since successors of an element of $S_F$ or predecessors of an element of $S_P$  in $\mathcal{N}$ might not be present in $\mathcal{N}{\upharpoonright}_X$.

\begin{defi}
Let $\mathcal{N} = (G,L,S_F,S_P)$ be a prenetwork and let $X \subseteq \underline{G}$. We define:
$$\upgen{\mathcal{N}}{X} = X \cup \{v \in \underline{G} \mid u \edge{G*} v  \text{ for some } u \in X\}$$ and 
$$\downgen{\mathcal{N}}{X} = X \cup \{v \in \underline{G} \mid v \edge{G*} u  \text{ for some } u \in X\}$$
We write $\upgen{\mathcal{N}}{u}$ and $\downgen{\mathcal{N}}{u}$ rather than $\upgen{\mathcal{N}}{\{u\}}$ and $\downgen{\mathcal{N}}{\{u\}}$, respectively.
\end{defi}

\begin{defi}
Let $\mathcal{N} = (G,L,S_F,S_P)$, $\mathcal{N}' = (G',L',S_F',S_P')$ be  $\Sigma$-networks, and let $u$ be a node in $\underline{G} \cap \underline{G}'$. We write $\mathcal{N} \equp{u} \mathcal{N}'$ if: $$\mathcal{N} \setminus (\upgen{\mathcal{N}}{u}) = \mathcal{N}' \setminus(\upgen{\mathcal{N}'}{u}).$$ Similarly we write $\mathcal{N} \eqdown{u} \mathcal{N}'$ if: 
$$\mathcal{N} \setminus (\downgen{\mathcal{N}}{u}) = \mathcal{N}' \setminus(\downgen{\mathcal{N}'}{u}).$$
\end{defi}

\begin{defi}
Given $\mathcal{N} \subseteq \mathcal{N}'$, we say that $\mathcal{N}$ is \emph{upwards cofinal} in $\mathcal{N}'$ if whenever $u \in \mathcal{N}$ and $u \reach{G'} v$ then $v \in \mathcal{N}$. Similarly we say that $\mathcal{N}$ is \emph{downwards cofinal} in $\mathcal{N}'$ if whenever $u \in \mathcal{N}$ and $v \reach{G'} u$ then $v \in \mathcal{N}$. 
\end{defi}

The following proposition establishes a kind of amalgamation property for \emph{finite and anticonfluent networks}:

\begin{prop}
\label{p:amalgamation}
Let $\mathcal{N} = (G,L,S_F,S_P)$ be a finite and anticonfluent  $\Sigma$-network, let $u_1,...u_n$ be nodes in $G$ and let $\mathcal{N}_1,...,\mathcal{N}_n$ be finite  and anticonfluent networks such that: 
\begin{itemize}
\item  $\mathcal{N} \sqsubseteq \mathcal{N}_i$ for each $i \in \{1,...,n\}$,
\item the sets $\upgen{\mathcal{N}_1}{u_1},...,\upgen{\mathcal{N}_n}{u_n}$ are pairwise disjoint,  
\item  $\mathcal{N} \equp{u_i} \mathcal{N}_i$ for each $i \in \{1,...,n\}$, and
\item $\mathcal{N}$ is downwards cofinal in $\mathcal{N}_i$ for each $i \in \{1,...,n\}$.
\end{itemize} 
Then there exists a finite  and anticonfluent $\Sigma$-network $\mathcal{N}'$ with $\mathcal{N}_i \sqsubseteq \mathcal{N}'$ for each $i \in \{1,...,n\}$, $\mathcal{N}$ is downwards cofinal  in $\mathcal{N}'$, and such that: 
$$\simup{\mathcal{N}}{\mathcal{N}'}{\{u_1,...,u_n\}}$$
Similarly, let $\mathcal{N} = (G,L,S_F,S_P)$ be a finite and anticonfluent  $\Sigma$-network, let $u_1,...u_n$ be nodes in $G$ and let $\mathcal{N}_1,...,\mathcal{N}_n$ be finite  and anticonfluent networks such that: 
\begin{itemize}
\item  $\mathcal{N} \sqsubseteq \mathcal{N}_i$ for each $i \in \{1,...,n\}$,
\item the sets $\downgen{\mathcal{N}_1}{u_1},...,\downgen{\mathcal{N}_n}{u_n}$ are pairwise disjoint,  
\item  $\mathcal{N} \eqdown{u_i} \mathcal{N}_i$ for each $i \in \{1,...,n\}$, and
\item $\mathcal{N}$ is upwards cofinal in $\mathcal{N}_i$ for each $i \in \{1,...,n\}$.
\end{itemize}
Then there exists a finite  and anticonfluent $\Sigma$-network $\mathcal{N}'$ with $\mathcal{N}_i \sqsubseteq \mathcal{N}'$ for each $i \in \{1,...,n\}$, $\mathcal{N}$ is upwards cofinal in $\mathcal{N}'$, and such that: 
$$\simdown{\mathcal{N}}{\mathcal{N}'}{\{u_1,...,u_n\}}$$
\end{prop}

\begin{proof}
We focus on the first part of the proposition, where the assumption is that the sets $\upgen{\mathcal{N}_1}{u_1},...,\upgen{\mathcal{N}_n}{u_n}$ are pairwise disjoint, $\mathcal{N} \sqsubseteq \mathcal{N}_i$ for each $i \in \{1,...,n\}$,   $\mathcal{N} \equp{u_i} \mathcal{N}_i$  for each $i \in \{1,...,n\}$ and $\mathcal{N}$ is downwards cofinal in $\mathcal{N}_i$    for each $i \in \{1,...,n\}$. The other case is handled in an analogous manner. 
We set $\mathcal{N}' = \mathcal{N}_1 \cup ... \cup \mathcal{N}_n$, which is defined because the assumptions entail that $L_i(u) = L_j(u) = L(u)$ if $u \in \underline{G}_i \cap \underline{G}_j$ for $i,j \in \{1,...,n\}$ (note that  $u \in \underline{G}_i \cap \underline{G}_j$ for $i \neq j$ implies $u \in \underline{G}$!). Clearly this network is finite, and $\mathcal{N}$ is downwards cofinal in $\mathcal{N}'$ since it is downwards cofinal in each $\mathcal{N}_i$. The situation for the special case of two networks $\mathcal{N}_1,\mathcal{N}_2$ extending $\mathcal{N}$ and with $\mathcal{N}' = \mathcal{N}_1 \cup \mathcal{N}_2$ is depicted in Figure \ref{f:amalg}, where each shaded area shows: (A) the whole network $\mathcal{N}'$, (B) the subnetwork $\mathcal{N}$, (C) the set $\upgen{\mathcal{N}_1}{u_1}$, (D) the set $\upgen{\mathcal{N}_2}{u_2}$, (E) the network $\mathcal{N}_1$, (F) the network $\mathcal{N}_2$.
\begin{figure}[h]
\begin{center}
\setlength{\unitlength}{0.5cm}
\begin{picture}(16,20)(0,0)
\put(2,0){\includegraphics[width = 7cm]{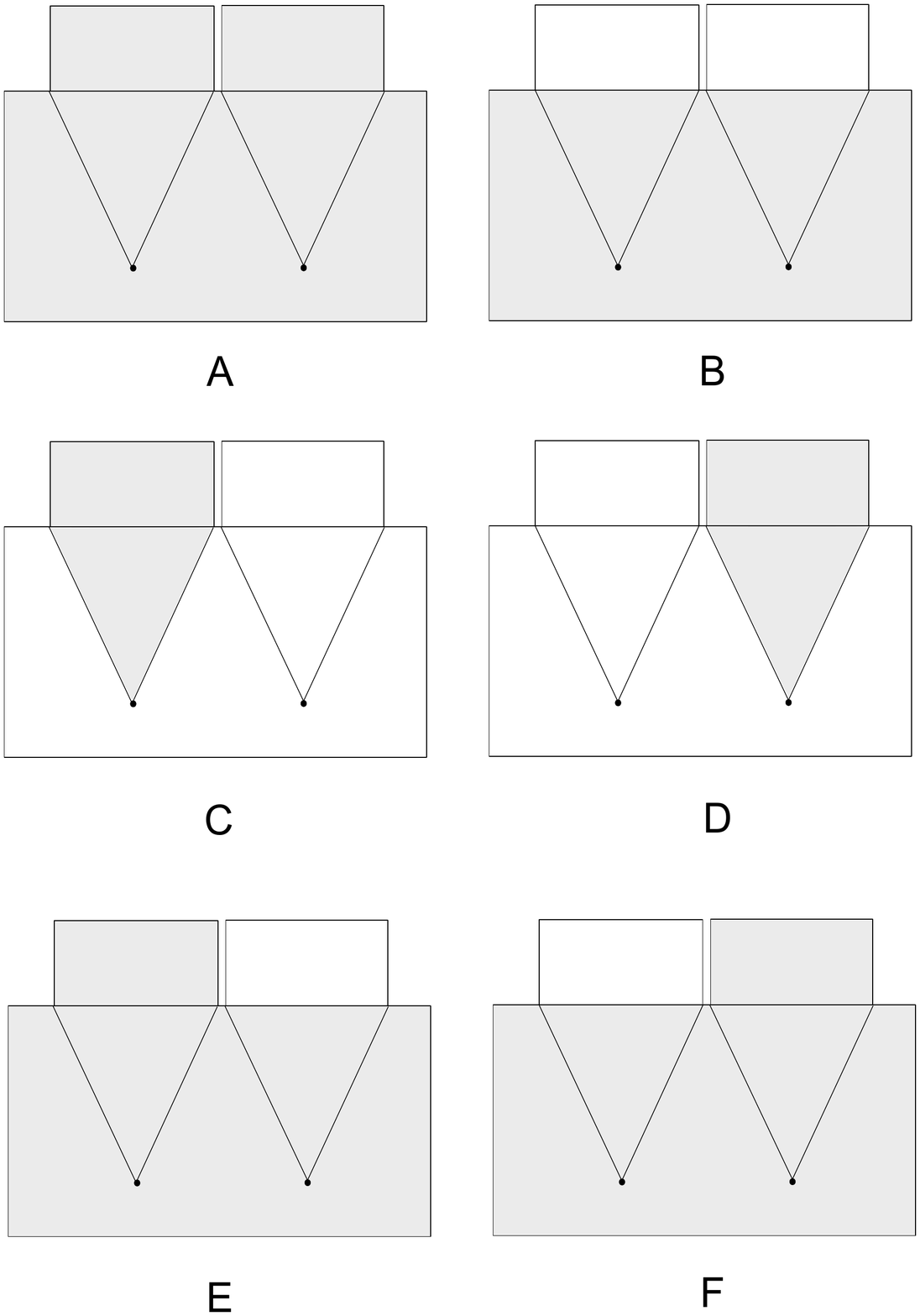}}
\end{picture}
\end{center}
\caption{\label{f:amalg} Amalgamation of networks}
\end{figure}

We have to check that the following claims all hold:
\begin{enumerate}
\item $\mathcal{N}'$ is anticonfluent,
\item $\mathcal{N}_i \sqsubseteq \mathcal{N}'$ for each $i \in \{1,...,n\}$,
\item $\simup{\mathcal{N}}{\mathcal{N}'}{\{u_1,...,u_n\}}$.
\end{enumerate}
For the proof of item (1), suppose $v_1,v_2,v_3,v_4$ are nodes in $\mathcal{N}'$ with $v_1 \reach{G'} v_2, v_1 \reach{G'} v_3, v_2 \reach{G'} v_4$ and $v_3 \reach{G'} v_4$, where $v_2 \neq v_3$. If both $v_2$ and $v_3$ belong to $\mathcal{N}$, then $v_1$ does too since $\mathcal{N}$ is downwards cofinal in $\mathcal{N}'$. Let $i$ be an index with $v_4$ belonging to $\mathcal{N}_i$. Since $\mathcal{N} \sqsubseteq \mathcal{N}_i$, all four nodes $v_1,v_2,v_3,v_4$ belong to $\mathcal{N}_i$, and this contradicts the assumption that $\mathcal{N}_i$ was anticonfluent. If one of $v_2,v_3$ belongs to $\mathcal{N}$ but not the other, say only $v_2$ belongs to $\mathcal{N}$, then pick $i$ with $v_3$ belonging to $\mathcal{N}_i$. Since $v_3$ is not in $\mathcal{N}$ and $\mathcal{N} \equp{u_i} \mathcal{N}_i$ we have $u_i \reach{G_i} v_3$, i.e. $v_3 \in \upgen{\mathcal{N}_i}{u_i}$. Since $v_3 \reach{G'} v_4$ and $\mathcal{N}$ is downwards cofinal in each $\mathcal{N}_j$, $v_4$ is not in $\mathcal{N}$. So we find some $\mathcal{N}_j$ with $v_4 \in \upgen{\mathcal{N}_j}{u_j}$. Since $v_3 \reach{G'} v_4$, by the construction of $\mathcal{N}'$ and the assumption that the sets $\upgen{\mathcal{N}_1}{u_1},...,\upgen{\mathcal{N}_n}{u_n}$ are pairwise disjoint, clearly this can only be the case if $i = j$. So we have $v_4$ belonging to $\mathcal{N}_i$ as well, and $v_1,v_2$ belong to $\mathcal{N}_i$ since they both belong to $\mathcal{N}$. Again we reach a contradiction since $\mathcal{N}_i$ was assumed to be anticonfluent. If neither $v_2$ nor $v_3$ belong to $\mathcal{N}$, then there must be $i,j$ with $v_2 \in \upgen{\mathcal{N}_i}{u_i}$ and $v_3 \in \upgen{\mathcal{N}_j}{v_j}$. Since $v_2 \reach{G'} v_4$ and $v_2 \reach{G'} v_4$, it again follows by similar reasoning as before that $i = j$ and $v_4 \in \mathcal{N}_i$. Furthermore, since $v_1 \reach{G'} v_2 $ and  $v_1 \reach{G'} v_3$, it is clear that $v_1$ belongs to $\mathcal{N}_i$ also: otherwise, it does not belong to $\mathcal{N}$, hence belongs to some set $\upgen{\mathcal{N}_j}{v_j}$  and again we get $i = j$, contradiction. So again, all four nodes $v_1,v_2,v_3,v_4$ belong to $\mathcal{N}_i$, and we get a contradiction. This shows that $\mathcal{N}'$ is anticonfluent.

Item (2) follows because the successors and predecessors of a node $v$ in $\mathcal{N}_i$ are the same as they are in $\mathcal{N}'$. Finally, item (3) is obvious from the assumption that  $\mathcal{N} \equp{u_i} \mathcal{N}_i$  for each $i \in \{1,...,n\}$.
\end{proof}

\subsection{Defects}

As usual in networks based arguments, our model construction will produce an infinite chain of finite approximating networks, where each network in the chain removes all the \emph{defects} of the previous one. We will introduce two types of defects, corresponding to modalities and fixpoint formulas. 

\begin{defi}
A \emph{$\fdia$-defect} in a network $\mathcal{N}$  is a node $u \notin S_F$. 
Similarly, a  \emph{$\bdia$-defect}  in a network $\mathcal{N}$  is a node $u \notin S_P$. 
\end{defi}

Defining defects corresponding to fixpoint formulas is more involved.

\begin{defi}
Suppose $\psi(x,\vec{q}) \in \Gamma$ is a potential $\Sigma$-deferral where $\psi(x,\vec{q})$ is a subformula of some formula  $\chi(x,\vec{q}) \in \Gamma$ with $\sharp_\chi \vec{\theta} \in \Sigma$.
Then $\psi(x,\vec{q})$ is called a $\sharp_\chi \vec{\theta}$-\emph{deferral} at $u$ in a network $\mathcal{N}$ if $\psi(\sharp_\chi \vec{\theta}, \vec{\theta}) \in L(u)$. In particular, the variable $x$ counts a  $\sharp_\chi \vec{\theta}$-deferral at $u$  if $\sharp_\chi \vec{\theta} \in L(u)$, and the formula $\sharp_\chi \vec{\theta}$ is then  called a \emph{focus} at $u$. 
\end{defi}

\begin{defi}
\label{d:finished}
Let $d$ be a $\sharp_\chi\vec{\theta}$-deferral at $u$ in $\mathcal{N} = (G,L,S_F,S_P)$ where $\chi(x,\vec{q})$ is a disjunctive fixpoint connective. We say that $d$ is \emph{finished in $k$ steps} at $u$ in $\mathcal{N}$ if:
\begin{itemize}
\item $d = x$ and $\chi(\bot,\vec{\theta}) \in L(u)$, or
\item $d = x$, $k = h + 1$ and the deferral $\chi(x,\vec{q})$ is finished in $h$ steps at $u$ in $\mathcal{N}$, or
\item $d = \psi_1(x,\vec{q}) \vee \psi_2(x,\vec{q})$ and for some $i \in \{1,2\}$, $\psi_i(\sharp_\chi \vec{\theta},\vec{\theta}) \in L(u)$ and the formula $\psi_i(x,\vec{q})$ is either not a potential $\sharp_\chi\vec{\theta}$-deferral or it is finished in $k$ steps, or
\item $d = \gamma \wedge \psi(x,\vec{q})$ where $x$ does not appear in $\gamma$,  $\psi(\sharp_\chi \vec{\theta},\vec{\theta}) \in L(u)$ and  $\psi_i(x,\vec{q})$ is either not a $\sharp_\chi\vec{\theta}$-deferral or it is finished in $k$ steps, and $\gamma \in L(u)$, or
\item  $d = \fnabla \{\psi_1(x,\vec{q}),...,\psi_n(x,\vec{q})\}$, $u \in S_F$, and there is a full relation $Z$ between the set of successors of $u$ in $G$ and the set of formulas $\{\psi_1(x,\vec{q}),...,\psi_n(x,\vec{q})\}$ such that:
\begin{enumerate}
\item if $v Z \gamma(x,\vec{q})$ then $\gamma(\sharp_\chi \vec{\theta}, \vec{\theta}) \in L(v)$, and
\item if $v Z \gamma(x,\vec{q})$ and $\gamma(x,\vec{q})$ is a $\sharp_\chi \vec{\theta}$-deferral, then this deferral is finished in $k$ steps at $v$.
\end{enumerate}
\item  $d = \bnabla \{\psi_1(x,\vec{q}),...,\psi_n(x,\vec{q})\}$, $u \in S_P$, and there is a full relation $Z$ between the set of predecessors of $u$ in $G$ and the set of formulas $\{\psi_1(x,\vec{q}),...,\psi_n(x,\vec{q})\}$ such that:
\begin{enumerate}
\item if $v Z \gamma(x,\vec{q})$ then $\gamma(\sharp_\chi \vec{\theta}, \vec{\theta}) \in L(v)$, and
\item if $v Z \gamma(x,\vec{q})$ and $\gamma(x,\vec{q})$ is a $\sharp_\chi \vec{\theta}$-deferral, then this deferral is finished in $k$ steps at $v$.
\end{enumerate}
\end{itemize} 
We say that a deferral $d$ is \emph{eventually finished} at $u$ in the network $\mathcal{N}$ if for some $k < \omega$, it is finished in $k$ steps. If this is not the case then we call $d$ a $\mu$-\emph{defect} of $u$. If $d$ is eventually finished then the smallest $k$ witnessing this is called the \emph{timeout} of $d$ at $u$ in $\mathcal{N}$. 

We say that the focus $\sharp_\chi \vec{\theta} \in L(u)$ is eventually finished at $u$ if the $\sharp_\chi \vec{\theta}$-deferral $x$ is eventually finished at $u$.
\end{defi}

%\begin{defi}
%A network $\mathcal{N}$ is said to be \emph{originating} at $u$ if the graph $G$ is a tree rooted at $u$. We say that $\mathcal{N}$ \emph{ends} in $u$ if the converse of the graph $G$ is a tree rooted at $u$. 
%\end{defi}

An important property of the subnetwork relation is that any deferral finished in a subnetwork is finished in the larger network also:

\begin{prop}
\label{p:stayfinished}
Suppose $\mathcal{N} \sqsubseteq \mathcal{N}'$, and let $\chi(x,\vec{q})$ be a disjunctive fixpoint connective. Then any $\sharp_\chi \vec{\theta}$-deferral $d$ at a node $u \in \mathcal{N}$ that is eventually finished in the network  $\mathcal{N}$ is eventually finished in $\mathcal{N}'$ also.
\end{prop}

\begin{proof}
By nested induction on the timeout of $d$ at $u$ in $\mathcal{N}$ and the complexity of the deferral $d$, we prove that each eventually finished deferral $d$ has the same timeout in $\mathcal{N}'$ as it had in $\mathcal{N}$.  If $d = x$ and $d$ has timeout $0$ then $\chi(\bot,\vec{\theta}) \in L(u) = L'(u)$, so $d$ is finished with timeout $0$ in $\mathcal{N}'$ as well. If $d = x$ and $d$ has timeout $k + 1$ then the deferral $\chi(x,\vec{q})$ has timeout $k$ in $\mathcal{N}$, and we can apply the induction hypothesis to the timeout $k$. If $d = d_1 \vee d_2$ and has timeout $k$ in $\mathcal{N}$ then either  one of the disjuncts belongs to $L(u)$ and is not a deferral (this case is trivial), or $d_1$ has timeout $k $ in $\mathcal{N}$ or $d_2$ has timeout $k$ in $\mathcal{N}$, and we can apply the induction hypothesis to either $d_1$ or $d_2$. The easy case of $d = \gamma \wedge \psi(x,\vec{q})$ where $x$ does not appear in $\gamma$ is left to the reader.

Finally, suppose $d = \fnabla \{\psi_1(x,\vec{q}),...,\psi_n(x,\vec{q})\}$. If $d$ is eventually finished at $u $ in $\mathcal{N}$, then $u \in S_F$  and there is a full relation $Z$ between successors of $u$ in $G$ and the formulas $\{\psi_1(x,\vec{q}),...,\psi_n(x,\vec{q})\}$ such that:
 if $v Z \gamma(x,\vec{q})$ then $\gamma(\sharp_\chi \vec{\theta}, \vec{\theta}) \in L(v)$, and
if $v Z \gamma(x,\vec{q})$ and $\gamma(x,\vec{q})$ is a $\sharp_\chi \vec{\theta}$-deferral, then this deferral is finished in $k$ steps at $v$.
 But since $u \in S_F$ and $\mathcal{N} \sqsubseteq \mathcal{N}'$, it follows by the definition of a subnetwork that the successors of $u$ in $G'$ are the same as in $G$. So $Z$ is still a full relation between the successors of $u$ in $G'$ and the formulas $\{\psi_1(x,\vec{q}),...,\psi_n(x,\vec{q})\}$, and the required conditions (1) and (2) from Definition \ref{d:finished} hold (the second condition due to the induction hypothesis applied to the formulas $\psi_1(x,\vec{q}),...,\psi_n(x,\vec{q})$). The case of $d = \bnabla \{\psi_1(x,\vec{q}),...,\psi_n(x,\vec{q})\}$ is proved in the same way, using $u \in S_P$.
\end{proof}

\begin{defi}
A network is said to be \emph{perfect} if it has no defects.
\end{defi}

From  perfect networks, we can construct satisfying models for consistent formulas:

\begin{defi}
The Kripke model $\bbS_\mathcal{N} = (W_\mathcal{N},R_\mathcal{N},V_\mathcal{N})$ induced by a $\Sigma$-network $\mathcal{N} = (G,L,S_F,S_P)$ is defined by setting $W_\mathcal{N} = \underline{G}$, $u R_\mathcal{N} v$ iff $u \edge{G} v$, and set $$V_{\mathcal{N}}(p) = \{u \in W_\mathcal{N} \mid p \in L(u)\}$$ for $p \in \Sigma$, $V(p) = \emptyset$ otherwise.
\end{defi}

\begin{prop}
\label{p:ismodel}
Suppose $\Sigma$ is a finite Fischer-Ladner closed set of formulas in $\languagesym_\Gamma(\Prop)$, where $\Gamma$ is a sef of disjunctive fixpoint connectives. If $\mathcal{N} = (G,L,S_F,S_P)$ is a perfect $\Sigma$-network and $u \in \underline{G}$, then $\bbS_\mathcal{N},u \Vdash \bigwedge L(u)$.
\end{prop} 

 \begin{proof}
We prove by a nested induction on the fixpoint nesting depth and the  complexity of a formula $\varphi \in \Sigma$ that $\varphi \in L(u)$ iff $\bbS_\mathcal{N},u \Vdash \varphi$. The cases for propositional variables and boolean connectives are completely standard. The cases for modal operators is straightforward, using the fact that $\mathcal{N}$ has no $\fdia$-defects or $\bdia$-defects.

Now suppose that $\sharp_\chi \vec{\theta} \in L(u)$. Since $\mathcal{N}$ has no $\mu$-defects,  the deferral $\chi(x,\vec{q})$ has a timeout $k$. We prove, by simultaneous induction on the timeout $k$ of an eventually finished $\sharp_\chi \vec{\theta}$-deferral $\psi(x,\vec{q})$ at a node $u$, and on the complexity of the deferral, that $\bbS_\mathcal{N}, u\Vdash \psi(\delta^k,\vec{\theta})$, where we write $\delta^k = \chi^k(\bot,\vec{\theta})$. That is, $\delta^k$ is defined inductively as $\delta^0 = \chi(\bot,\vec{\theta})$ and $\delta^{k + 1} = \chi(\delta^k,\vec{\theta})$. The desired conclusion then follows from Proposition \ref{p:approximants}. We focus on the case where $\chi$ is forward-looking since the other case is proved in the same manner. 

The base case is where $d = x$ and $k = 0$. Then $\chi(\bot,\vec{\theta}) \in L(u)$, and since this formula has lower fixpoint depth than $\sharp_\chi \vec{\theta}$ the main induction hypothesis applies and we get $\bbS_\mathcal{N},u \Vdash \chi(\bot,\vec{\theta})$, i.e. $\bbS_\mathcal{N},u \Vdash \delta^0$ as required. 

Now suppose the induction hypothesis holds for all timeouts smaller than $k$ and for all deferrals of lower complexity than $d$. We prove that the hypothesis holds also for $d$ with the timeout $k$. The induction step for disjunctions is entirely straightforward, and so is the case for $\gamma \wedge \psi(x,\vec{\theta})$ where $x $ does not occur in $\gamma$. The case where $d = x$ and $k = 0$ has already been taken care of. The case for $d = x$ and $k = k' +1$ for some $k'$ is handled by noting that $\delta^{k} = \chi(\delta^{k'},\vec{\theta})$, and we apply the induction hypothesis on $k'$  to the deferral $\chi(x,\vec{q})$ to get $\bbS_\mathcal{N},u \Vdash    \chi(\delta^{k'},\vec{\theta})$, i.e.  $\bbS_\mathcal{N},u \Vdash    \delta^{k}$ as required.

For the case where $d = \fnabla \{\psi_1(x,\vec{q}),...,\psi_n(x,\vec{q})\}$, we have $u \in S_F$,  and there is a full relation $Z$ between successors of $u$ and the formulas $\{\psi_1(x,\vec{q}),...,\psi_n(x,\vec{q})\}$ such that:
\begin{enumerate}
\item if $v Z \gamma(x,\vec{q})$ then $\gamma(\sharp_\chi \vec{\theta}, \vec{\theta}) \in L(v)$, and
\item if $v Z \gamma(x,\vec{q})$ and $\gamma(x,\vec{q})$ is a $\sharp_\chi \vec{\theta}$-deferral, then this deferral is finished in $k$ steps at $v$.
\end{enumerate}
Now if $v Z \gamma(x,\vec{q})$ and $\gamma(x,\vec{q})$ is not a deferral, the variable $x$ does not appear in $\gamma(x,\vec{q})$ and so  the formula $\gamma(\sharp_\chi \vec{\theta}, \vec{\theta})$ has smaller fixpoint nesting depth than $\sharp_\chi \vec{\theta}$, so  by the induction hypothesis on the formula $\gamma(\delta^{k}, \vec{\theta})$  we have $\bbS_\mathcal{N},v \Vdash \gamma(\delta^{k},\vec{\theta})$. On the other hand, if $\gamma(x,\vec{q})$ is a deferral then this deferral is finished in $k$ steps at $v$, and the induction hypothesis on the complexity of deferrals applies to the deferral $\gamma(x,\vec{q})$, so again we have $\bbS_\mathcal{N},v \Vdash \gamma(\delta^{k},\vec{\theta})$. Since $Z$ was a full relation between the successors of $u$ and the formulas $\{\psi_1(x,\vec{q}),...,\psi_n(x,\vec{q})\}$ it follows that:
$$\bbS_\mathcal{N},v \Vdash \fnabla \{\psi_1(\delta^{k},\vec{\theta}),...,\psi_n(\delta^{k},\vec{\theta}) $$
as required. 
The case where $d = \bnabla \{\psi_1(x,\vec{q}),...,\psi_n(x,\vec{q})\}$ is handled in the same manner.

Conversely, suppose that $\sharp_\chi \vec{\theta} \notin L(u)$. Then $u \notin F = \{v \in G \mid \sharp_\chi \vec{\theta} \in L(v)\}$, and so it suffices to show that $F$ is a prefixpoint for the map $\tset{\chi(-,\vec{\theta})}_{\bbS_\mathcal{N}}$. Since $F =  \{v \in G \mid \chi(\sharp_\chi \vec{\theta},\vec{\theta}) \in L(v)\}$ -- because the formula $\sharp_\chi \vec{\theta} \leftrightarrow  \chi(\sharp_\chi \vec{\theta},\vec{\theta})$ is provable -- it suffices to show by induction on the complexity of a $\sharp_\chi \vec{\theta}$-deferral $\psi(x,\vec{q})$ that for all $v \in \underline{G}$, $\psi(\sharp_\chi \vec{\theta},\vec{\theta}) \in L(v)$  iff $\bbS_\mathcal{N}[x \mapsto F],v \Vdash \psi(x,\vec{\theta})$.  For the case where $\psi(x,\vec{q}) = x$, we have $\psi(\sharp_\chi \vec{\theta},\vec{\theta}) = \sharp_\chi \vec{\theta}$ and the result follows immediately by definition of the valuation of $x$ as the set $\{v \in G \mid \chi(\sharp_\chi \vec{\theta},\vec{\theta}) \in L(v)\}$  in the model $\bbS_\mathcal{N}[x \mapsto F]$. The rest of the induction now proceeds by a straightforward induction on the complexity of $\psi(x,\vec{q})$, using the main induction hypothesis on the formulas $\vec{\theta}$, so we omit the details.
\end{proof}

The rest of the completeness proof is devoted to showing how to remove defects in a network. We start with the easy case:

\begin{prop}
\label{p:removefpdefects}
For every node $u$ in a finite anticonfluent network $\mathcal{N}$, there exists a finite and anticonfluent network $\mathcal{N}'$ with $\mathcal{N} \sqsubseteq \mathcal{N}'$ in which $u$ is not an $\fdia$-defect. Similarly, for every node $u$ in a finite anticonfluent network $\mathcal{N}$, there exists a finite and anticonfluent network $\mathcal{N}'$ with $\mathcal{N} \sqsubseteq \mathcal{N}'$ in which $u$ is not a $\bdia$-defect.
\end{prop}

\begin{proof}
Let $\mathcal{N} = (G,L,S_F,S_P)$.
We focus on the first part of the proposition since the second part is proved by essentially the same argument. If $u \in S_F$ then $u$ is already not an $\fdia$-defect, so we assume that $u \notin S_F$. Pick a maximal consistent set $\Gamma$ such that $L(u) = \Gamma \cap \Sigma$. For each formula $\fdia \varphi \in L(u)$, pick a successor $\Theta_\varphi$ of $\Gamma$ in the canonical model containing $\varphi$. Pick a distinct fresh object $v^\varphi_j$ of for each $\fdia \varphi \in L(u)$ and each $j \in \{1,...,\dfrl\}$, and define the graph $G'$ by taking $\underline{G'}$  to consist of the nodes in $\underline{G}$ together with all the fresh objects $v^\varphi_j$, and setting $w \edge{G'} w'$ iff either $w,w' \in \underline{G}$ and $w \edge{G} w'$, or $w = u$ and $w' = v^\varphi_j$ for some $\varphi,j$.   We set  $L'(v^\varphi_j) = \Theta_\varphi \cap \Sigma$ for $\fdia \varphi \in L(u)$ and $j \in \{1,...,\dfrl\}$, and for $w \in \underline{G}$ we set $L'(w) = L(w)$. Finally, we set $S_P' = S_P$, and $S_F' = S_F \cup \{u\}$. It is straightforward to check that $\mathcal{N}'$ is a network and that $\mathcal{N} \sqsubseteq \mathcal{N}'$, and of course $u$ is not a $\fdia$-defect of $\mathcal{N}'$ since $u \in S_F'$.
\end{proof}

\section{Using the induction rule to build networks}

The difficult part of the construction of a perfect network is showing how to finish deferrals. The first and most important step towards solving this problem is given by the following result, the proof of which is reminiscent of a well known completeness proof for propositional dynamic logic $\mathtt{PDL}$ (see \cite{blac:moda01}).

\begin{prop}
\label{p:inductionlemma}
Suppose $\Sigma$ is a finite Fischer-Ladner closed set of formulas in $\languagesym_\Delta(\Prop)$, where $\Delta$ is a sef of guarded and disjunctive fixpoint connectives. Let $A$ be a $\Sigma$-atom containing $\sharp_\chi \vec{\theta}$. If $\chi$ is forward-looking then there exists a finite network $\mathcal{N}$ containing a node $u$ with $L(u) = A$, such that $G$ is a tree rooted at $u$ (hence the network is anticonfluent), $S_P = \emptyset$, and in which the focus $\sharp_\chi \vec{\theta}$ is eventually finished at $u$. Similarly, if $\chi$ is backward-looking then there exists a finite  network $\mathcal{N}$ containing a node $u$ with $L(u) = A$,  such that $G^{\mathit{op}}$ is a tree rooted at $u$ (hence the network is anticonfluent), $S_F = \emptyset$, and in which the focus $\sharp_\chi \vec{\theta}$ is eventually finished at $u$. 
\end{prop}

\begin{proof}
We focus on the case where $\chi$ is forward-looking, since the other case is proved in the same way.

Let $\delta$ be the disjunction of all formulas of the form $\bigwedge A$ such that:
\begin{itemize}
\item $A$ is a $\Sigma$-atom,
 \item $A$ contains $\sharp_\chi \vec{\theta}$,
\item there exists a finite  network $\mathcal{N}$ of which the underlying graph $G$ is a tree rooted at a node $u$, such that:
\item $L(u) = A$, and the focus $\sharp_\chi \vec{\theta}$ at $u$ is eventually finished in $\mathcal{N}$. 
\end{itemize}
Our aim is to prove $\vdash \chi(\delta,\vec{\theta}) \rightarrow \delta$;  we then get $\vdash \sharp_\chi \vec{\theta} \rightarrow \delta$ by the induction rule. Hence every atom $A$ containing $\sharp_\chi \vec{\theta}$ provably entails $\delta$, which means that the conjunction of such an atom must itself be one of the disjuncts of $\delta$: otherwise $A$ would be inconsistent with $\delta$ because any two distinct $\Sigma$-atoms are inconsistent with each other, so $A$ would be inconsistent, which contradicts the definition of an atom. By definition of $\delta$ this means that a suitable network satisfying the criteria stated in the proposition exists for the atom $A$, so we can conclude that this holds for any atom that contains $\sharp_\chi \vec{\theta}$, as required.

To prove the implication, we show that if an atom $A$ is consistent with the formula  $\chi(\delta,\vec{\theta})$, then there exists a  network $\mathcal{N} = (G,L,S_F,S_P)$ containing a node $u$ with $L(u) = A$, such that $G$ is a finite tree rooted at $u$, $S_P = \emptyset$, and the focus $\sharp_\chi \vec{\theta}$ at $u$ is eventually finished in $\mathcal{N}$.\footnote{Note that $\vdash \chi(\delta,\vec{\theta}) \rightarrow \sharp_\chi \vec{\theta}$: since clearly $\delta \vdash \sharp_\chi \vec{\theta}$ we get $\chi(\delta,\vec{\theta}) \vdash \chi(\sharp_\chi \vec{\theta},\vec{\theta})$, and the last formula is provably equivalent to $\sharp_\chi(\vec{\theta})$. So if $A$ is consistent with $\chi(\delta,\vec{\theta})$ then it is consistent with $\sharp_\chi \vec{\theta}$ and so since $\sharp_\chi \vec{\theta} \in \Sigma$ we get $\sharp_\chi \vec{\theta} \in A$.} It then follows that $\vdash \chi(\delta,\vec{\theta}) \rightarrow \delta$, for suppose otherwise: then there is a maximal consistent set $\Gamma$ containing $\chi(\delta,\vec{\theta})$ and $\neg \delta$. The atom $\Gamma \cap \Sigma$ is therefore consistent with $\chi(\delta,\vec{\theta})$, so there exists a network $\mathcal{N} = (G,L,S_F,S_P)$ containing a node $u$ with $L(u) = \Gamma \cap \Sigma$, such that $G$ is a finite tree rooted at $u$, and the focus $\sharp_\chi \vec{\theta}$ at $u$ is eventually finished in $\mathcal{N}$. But then $\bigwedge(\Gamma \cap \Sigma)$ must be one of the disjuncts in $\delta$, and  $\bigwedge(\Gamma \cap \Sigma) \in \Gamma$ so $\delta \in \Gamma$, contradicting  that $\neg\delta \in \Gamma$.  

We shall prove something more general: for every potential $\sharp_\chi \vec{\theta}$-deferral   $\psi(x,\vec{q})$ and for every atom $A$ such that the set $A \cup \{\psi(\delta,\vec{\theta})\}$ is consistent, there exists a  network $\mathcal{N} = (G,L,S_F,S_P)$ containing a node $u$ with $L(u) = A$, such that $G$ is a finite tree rooted at $u$, $S_P = \emptyset$, and such that the deferral $\psi(x,\vec{q})$ at $u$ is eventually finished in $\mathcal{N}$. We prove this by induction on the complexity of the formula $\psi(x,\vec{q})$.

If $\psi(x,\vec{q}) = x$ then the assumption amounts to saying that $A$ is consistent with $\delta$, which can only be the case if $\bigwedge A$ is one of the disjuncts of $\delta$. So in this case, it follows by definition of $\delta$ that a suitable network exists for $A$. 

If $\psi(x,\vec{q}) = \gamma_1(x,\vec{q}) \vee \gamma_2(x,\vec{q})$, then at least one of the formulas $\gamma_1(\delta,\vec{\theta})$ or  $ \gamma_2(\delta,\vec{\theta})$ must be consistent with $A$. Applying the induction hypothesis to this disjunct we obtain the required network for $A$, since if the $\sharp_\chi \vec{\theta}$-deferral $\gamma_i(x,\vec{q}) $ with $i \in \{1,2\}$ is eventually finished at $u$ in $\mathcal{N}$ then clearly so is  $\gamma_1(x,\vec{q}) \vee \gamma_2(x,\vec{q})$.

For the case of a potential deferral $\gamma \wedge \psi(x,\vec{q})$ where $x$ does not appear in $\gamma$, if $\gamma \wedge \psi(\delta,\vec{\theta})$ is consistent with $A$ then $\gamma \in A$ and $A$ is consistent with $\psi(\delta,\vec{\theta})$. So by the induction hypothesis we find an appropriate network $\mathcal{N}$ the underlying graph of which is a tree rooted at a node $u$ labelled with $A$, such that the deferral $\psi(x,\vec{q})$ is eventually finished at $u$ in $\mathcal{N}$. Since $\gamma \in L(u)$, the deferral $\gamma \wedge \psi(x,\vec{q})$ is also eventually finished at $u$ in $\mathcal{N}$. 

Now suppose $\psi(x,\vec{q}) = \fnabla \{\psi_1(x,\vec{q}),...,\psi_m(x,\vec{q})\}$, and suppose $A$ is consistent with the formula  $\fnabla \{\psi_1(\delta,\vec{\theta}),...,\psi_m(\delta,\vec{\theta})\}$. Let $\Gamma$ be a maximal consistent set extending the set $A$ and containing the formula $\fnabla \{\psi_1(\delta,\vec{\theta}),...,\psi_m(\delta,\vec{\theta})\}$. For each formula $\fdia \varphi \in A$, pick a successor $\Theta_\varphi$ of $\Gamma$ in the canonical model.  Note that every such successor contains one of the formulas $\psi_j(\delta,\vec{\theta})$, so let $$c : \{\varphi \mid \fdia \varphi \in A\} \to \{1,...,m\}$$ be a choice function such that $\psi_{c(\varphi)}(\delta,\vec{\theta}) \in \Theta_\varphi$ for each $\varphi$. Furthermore, for each $j \in \{1,...,m\}$, pick a successor $\Psi_j$ of $\Gamma$ in the canonical model with $\psi_j(\delta,\vec{\theta}) \in \Psi_j$, which exists since $\fdia\psi_j(\delta,\vec{\theta}) \in \Gamma $. 
Now, for each pair $(\varphi,i) \in \{ \varphi \mid \fdia \varphi \in A\} \times \{1,...,\dfrl\}$ we define a network $\mathcal{N}^\varphi_i$ as follows: if $c(\varphi) = \psi_j(x,\vec{q})$ is a potential $\sharp_\chi \vec{\theta}$-deferral, then we  define $\mathcal{N}^\varphi_i$ to be an arbitrarily chosen network $(G^\varphi_i,L^\varphi_i,(S_F)^\varphi_i,(S_P)^\varphi_i)$ such that  $G^\varphi_i$ is a finite tree rooted at some node $w^i_\varphi$, $L_i^\varphi(w^i_\varphi) = \Theta_\varphi \cap \Sigma$, $(S_P)^\varphi_i = \emptyset$ and the deferral $\psi_j(x,\vec{q})$ is eventually finished at $w_i^\varphi$ in $\mathcal{N}^\varphi_i$. Such a network exists by the induction hypothesis on $\psi_j(x,\vec{q})$, since $\psi_j(\delta,\vec{\theta}) \in \Theta_\varphi$ and so is consistent with $\Theta_\varphi \cap \Sigma$. Otherwise, we define $\mathcal{N}_i^\varphi = (G^\varphi_i,L^\varphi_i,(S_F)^\varphi_i,(S_P)^\varphi_i)$ by setting $G^\varphi_i$ to consist of a single node $w_i^\varphi$ and the empty set of edges, $(S_F)^\varphi_i = (S_P)^\varphi_i = \emptyset$, and $L_i^\varphi(w_i^\varphi) = \Theta_\varphi \cap \Sigma$.

Similarly, for each $j \in \{1,...,m\}$, we define a network $\mathcal{N}_j' = (G_j',L_j',(S_F)_j', (S_P)_j')$ by a case distinction: if $\psi_j(x,\vec{q})$ is a potential $\sharp_\chi \vec{\theta}$-deferral,  then we  define $\mathcal{N}_j'$ to be an arbitrarily chosen network such that  $G'_j$ is a finite tree rooted at some node $w_j'$, $L_j'(w'_j) = \Psi_j \cap \Sigma$, $(S_P)_j' = \emptyset$ and the deferral $\psi_j(x,\vec{q})$ is eventually finished at $w_j'$ in $\mathcal{N}_j'$. Such a network exists by the induction hypothesis on $\psi_j(x,\vec{q})$, since $\psi_j(\delta,\vec{\varphi}) \in \Psi_j$ and so is consistent with $\Psi_j \cap \Sigma$. Otherwise, we define $\mathcal{N}_j' $ by setting $G_j'$ to consist of a single node $w_j'$ and the empty set of edges, $(S_F)_j' = (S_P)_j' = \emptyset$, and $L_j'(w_j') = \Psi_j' \cap \Sigma$.

We can of course assume without loss of generality that the networks $\mathcal{N}_j'$ for $j \in \{1,...,m\}$ and $\mathcal{N}^\varphi_i$ for $\fdia \varphi \in A$ and $i \in \{1,...,\dfrl\}$ are pairwise disjoint. We form a new network $\mathcal{N}'' = (G'',L'',S_F'',S_P'')$ as follows: the set $\underline{G}''$ consists of all the nodes of each network $\mathcal{N}_j'$ for $j \in \{1,...,m\}$ and $\mathcal{N}^\varphi_i$ for $\fdia \varphi \in A$ and $i \in \{1,...,\dfrl\}$ together with a fresh node $u$, and we set $v \edge{G''} v'$ iff:
\begin{itemize}
\item  either for some $j \in \{1,...,m\}$ we have $v,v' \in \underline{G}_j'$ and $v \edge{G_j'} v'$, or
\item  for some $\fdia \varphi \in A$ and $i \in \{1,...,\dfrl\}$ we have $v,v' \in \underline{G}_i^\varphi$ and $v \edge{G_i^\varphi} v'$, or
\item  for some $j \in \{1,...,m\}$ we have  $v = u$ and $v' = w_j'$, or 
\item for some  $\fdia \varphi \in A$ and $i \in \{1,...,\dfrl\}$ we have $v = u$ and $v' = w_i^\varphi$. 
\end{itemize}
See Figure \ref{f:network} below.
\begin{figure}[h]
\begin{center}
\setlength{\unitlength}{0.5cm}
\begin{picture}(16,13)(0,0)
\put(-2.7,-14){\includegraphics[width = 10cm]{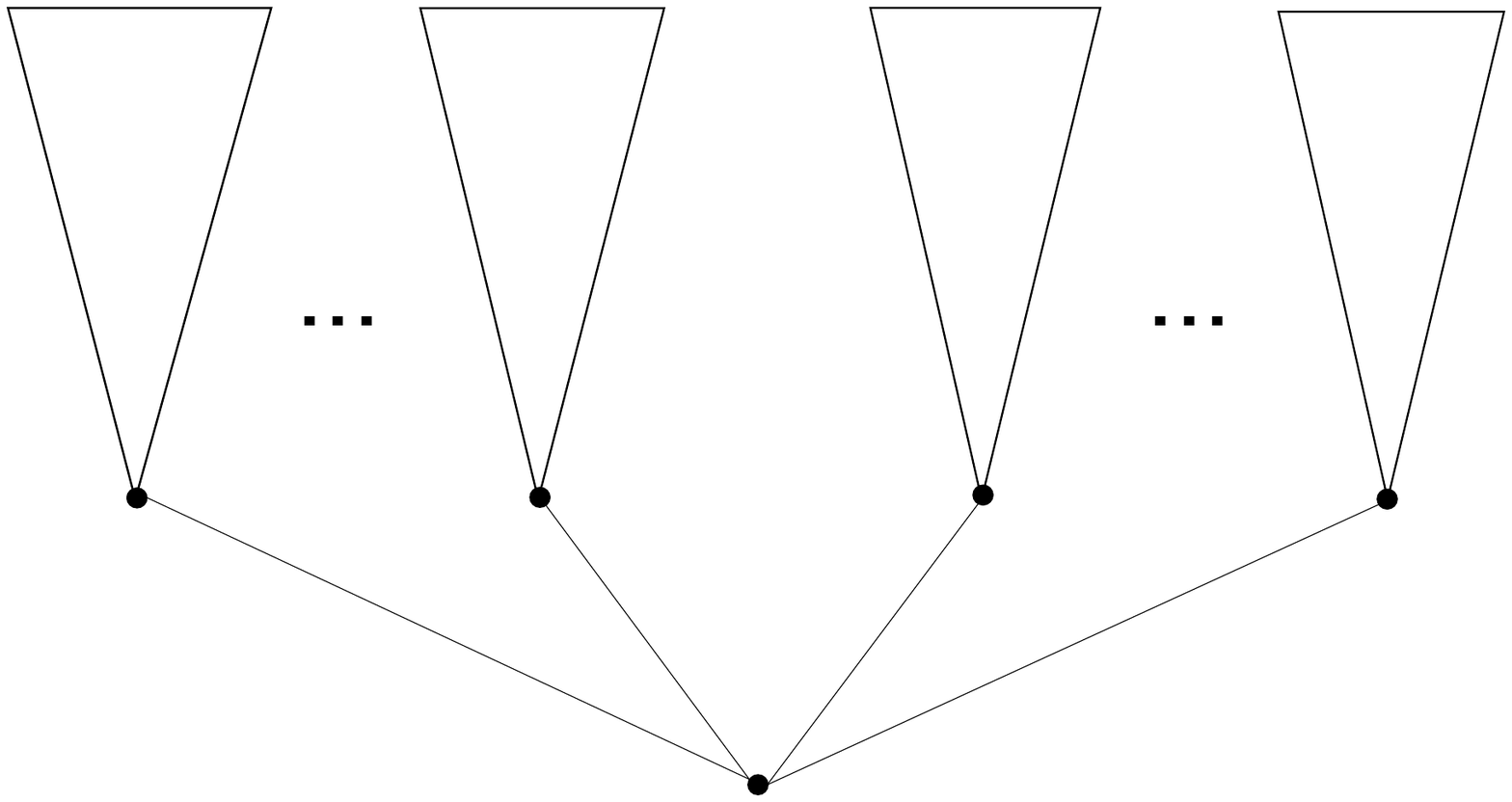}}
\put(6.8,-0.7){$u$}
\put(0.4,10.5){$\mathcal{N}'_1,...,\mathcal{N}_m'$}
\put(8,10.5){$\mathcal{N}_i^\varphi,\; \fdia \varphi \in A \; \& \; 1 \leq i \leq \dfrl$}
\end{picture}
\end{center}
\caption{\label{f:network} The graph $G''$ }
\end{figure}

The labelling $L''$ is defined by extending the labelling functions of the networks $\mathcal{N}_j'$ for $j \in \{1,...,m\}$ and $\mathcal{N}^\varphi_i$ for $\fdia \varphi \in A$ and $i \in \{1,...,\dfrl\}$ by putting  $L''(u) = A$. Finally we set:
 $$ S_F'' = \bigcup \{(S_F)_i^\varphi \mid \fdia \varphi \in A \; \& \; i \in \{1,...,\dfrl\}\} \cup \bigcup \{ S_j' \mid j \in \{1,...,m\}\} \cup \{u\}$$
and $S_P'' = \emptyset$.
\begin{claim}
\label{isnetwork}
The structure $\mathcal{N}''$ is a finite $\Sigma$-network and $G''$ is a tree rooted at $u$. 
\end{claim}
\begin{proof}[Proof of Claim \ref{isnetwork}]
It is immediate from the construction that $G''$ is a tree rooted at $u$. To check that $\mathcal{N}''$ is a $\Sigma$-network,
we check that the required conditions hold for each node $v$ in $\mathcal{N}''$. We have a few cases to consider:

\textbf{Case 1}: $v = u$. Since $u$ has no predecessors in $G''$ the condition for modalities $\bbox$ holds trivially. Suppose $\fbox \varphi \in L''(u) = A$. Then $\fbox \varphi \in \Gamma$, so $\varphi \in \Theta$ for each successor of $\Gamma$ in the canonical model, and since each successor of $u$ is labelled with $\Theta \cap \Sigma$ for some such successor $\Theta$ the conclusion follows. For the conditions involving the sets $S_F''$ and $S_P''$, the latter is vacuous since $u \notin S_P'' = \emptyset$. For the condition for $S_F''$, it suffices to note that $\varphi \in L''(w_i^\varphi) = \Theta_\varphi$ for each $\fdia \varphi \in A$ and $i \in \{1,...,\dfrl\}$.

\textbf{Case 2}: $v$ belongs to one of the networks $\mathcal{N}_j'$ for $j \in \{1,...,m\}$. We divide this case into subcases.

\textbf{Case 2a}: $j$ is such that $\psi_j(x,\vec{q})$ is not a potential $\sharp_\chi \vec{\theta}$-deferral. Then $v = w_j'$, and since $w_j'$ has no successors in $G''$ the clause for $\fbox$ is trivial. The only predecessor of $w_j'$ in $G''$ is $u$ which is labelled $A$, which is equal to $\Gamma \cap \Sigma$ and we recall that $\Gamma$ is a predecessor of $\Psi_j$ in the canonical model. Hence, if $\bbox \varphi \in L''(w_j') = L_j'(w_j') = \Theta_j \cap \Sigma$ then $\varphi \in L''(u) = \Gamma \cap \Sigma$, as required. The clause for the set $S_P'' = \emptyset$ is vacuous, and the clause for $S_F''$ holds because $w_j' \notin S_F''$.

\textbf{Case 2b}: $j$ is such that  $\psi_j(x,\vec{q})$ is a potential deferral. If $v = w_j'$ then the only successor of $w_j'$ is $u$, so the clause for $\bbox$ is treated as in Case 2a. Furthermore, all successors of $w_j'$ in $\mathcal{N}''$ are successors of $w_j'$ in the network $\mathcal{N}_j'$, so the clause for $\fbox$ holds in this case because $\mathcal{N}_j'$ was a $\Sigma$-network. If $v \neq w_j'$ then all successors (predecessors) of $v$ in $\mathcal{N}''$ are successors (predecessors) of $v$ in $\mathcal{N}_j'$,   so the clause for $\fbox$ ( for $\bbox$) holds  because $\mathcal{N}_j'$ was a $\Sigma$-network. The clause for the set $S_P'' = \emptyset$ is vacuous, and the clause for $S_F''$ holds because $\mathcal{N}_j'$ was a $\Sigma$-network.

\textbf{Case 3}: $v$ belongs to one of the networks $\mathcal{N}_i^\varphi$ for $i \in \{1,...,\dfrl\}$ and $\fdia \varphi \in A$. This case is similar to Case 2.
\end{proof}

\begin{claim}
\label{deferral}
The deferral $\fnabla \{\psi_1(\delta,\vec{\theta}),...,\psi_m(\delta,\vec{\theta})\}$ is eventually finished at the node $u$ in $\mathcal{N}'$.
\end{claim}
\begin{proof}[Proof of Claim \ref{deferral}]
We define a relation $$Z \subseteq (\{w_j' \mid 1 \leq j \leq m\} \cup \{w_i^\varphi \mid 1 \leq i \leq \dfrl \; \& \; \fdia \varphi \in A\}) \times \{\psi_1(x,\vec{q}),...,\psi_m(x,\vec{q})\}$$
by setting:
$$Z = \{(w_j,\psi_j(x,\vec{q})) \mid 1 \leq j \leq m\} \cup \{(w_i^\varphi,c(\varphi)) \}\mid 1 \leq i \leq \dfrl \; \& \; \fdia \varphi \in A\}$$
It follows from the construction of $\mathcal{N}''$ that this is a full relation between the successors of $u$ and the formulas $\{\psi_1(\delta,\vec{\theta}),...,\psi_m(\delta,\vec{\theta})\}$. If $(v,\psi_j(x,\vec{q})) \in Z$ then $L''(v)$ is of the form $\Theta \cap \Sigma$ where $\psi_j(\delta,\vec{\theta}) \in \Theta$. It is clear that $\delta \vdash \sharp_\chi \vec{\theta}$ so $\psi_j(\delta,\vec{\theta}) \vdash \psi_j(\sharp_\chi \vec{\theta},\vec{\theta})$. Hence $\psi_j(\sharp_\chi \vec{\theta},\vec{\theta}) \in \Theta$, and since $\psi_j(\sharp_\chi \vec{\theta},\vec{\theta}) \in \Sigma$ we get $\psi_j(\sharp_\chi \vec{\theta},\vec{\theta}) \in L''(v)$. Furthermore, it is clear from the construction that if $(v,\psi_j(x,\vec{q})) \in Z$ then the deferral $\psi_j(x,\vec{q})$ is eventually finished at $v$. We are here using the obvious fact that $\mathcal{N}_j' \sqsubseteq \mathcal{N}''$ for each $j \in \{1,...,m\}$ and $\mathcal{N}^\varphi_i \sqsubseteq \mathcal{N}''$ for each $i \in \{1,...,\dfrl\}$ and $\fdia \varphi \in A$, which means we can apply Proposition \ref{p:stayfinished}. 
\end{proof}
With the claims \ref{isnetwork} and \ref{deferral} in place, the proof is finished. 
\end{proof}

\begin{prop}
\label{p:stronginductionlemma}
Let $A$ be an atom containing $\psi(\sharp_\chi \vec{\theta},\vec{\theta})$ where $\psi(x,\vec{q})$ is a potential $\sharp_\chi \vec{\theta}$-deferral. If $\chi$ is forward-looking then there exists a finite  network $\mathcal{N}$ containing a node $u$ with $L(u) = A$, such that $G$ is a tree rooted at $u$, and in which the deferral $\psi(x,\vec{q})$ is eventually finished. Similarly, if $\chi$ is backward-looking then there exists a finite  network $\mathcal{N}$ containing a node $u$ with $L(u) = A$,  such that $G^{\mathit{op}}$ is a tree rooted at $u$, and in which the deferral $\psi(x,\vec{q})$ is eventually finished.
\end{prop} 
\begin{proof}
For the forward-looking case, we prove by induction on the complexity of a potential $\sharp_\chi \vec{\theta}$-deferral $\psi(x,\vec{q})$ that if an atom $A$ contains  the formula $\psi(\sharp_\chi \vec{\theta},\vec{\theta})$, then there exists a  network $\mathcal{N} = (G,L,S_F,S_P)$ containing a node $u$ with $L(u) = A$, such that $G$ is a finite tree rooted at $u$, $S_P = \emptyset$, and the deferral $\psi(x,\vec{\theta})$ at $u$ is eventually finished in $\mathcal{N}$. The base case of the induction where $\psi(x,\vec{q}) = x$ is taken care of by Proposition \ref{p:inductionlemma}, and the rest of the induction follows similar reasoning as in the proof of Proposition \ref{p:inductionlemma}.

The backwards-looking case is handled in a  similar manner. 
\end{proof}

\section{Construction of a perfect network}

We now have the tools in place to show how to construct a perfect network for any atom.

\begin{prop}
\label{p:onlyhead}
For every finite and anticonfluent network $\mathcal{N}$, there exists a finite and anticonfluent network $\mathcal{N}'$ with $\mathcal{N} \sqsubseteq \mathcal{N}'$, and such that every $\fdia$-defect in $\mathcal{N}'$ is a head node. 

Similarly, for every  finite and anticonfluent network $\mathcal{N}$, there exists a finite and anticonfluent network $\mathcal{N}'$ with $\mathcal{N} \sqsubseteq \mathcal{N}'$, and such that every $\bdia$-defect in $\mathcal{N}'$ is a tail node. 
\end{prop}

\begin{proof}
Again, we focus on the forward-looking case since the backward-looking case is proved by a similar argument. 
We prove that, for every finite and anticonfluent network $\mathcal{N}$ such that some $\fdia$-defects are not head nodes, there exists a finite and anticonfluent network $\mathcal{N}'$ with $\mathcal{N} \sqsubseteq \mathcal{N}'$, and such that the number of $\fdia$-defects in $\mathcal{N}'$ that are not head nodes is strictly smaller than the number of $\fdia$-defects in $\mathcal{N}$ that are not head nodes. This clearly implies the result, since there are only finitely many $\fdia$-defects in a finite network.

To prove the result, let $u$ be an $\fdia$-defect in a finite and anticonfluent network $\mathcal{N}$, such that $u$ is not a head node. We construct a (finite and anticonfluent) network $\mathcal{N}'$ by similar reasoning as in the proof of Proposition \ref{p:removefpdefects}: we pick a maximal consistent set $\Gamma$ such that $L(u) = \Gamma \cap \Sigma$, and for each formula $\fdia \varphi \in L(u)$, pick a successor $\Theta_\varphi$ of $\Gamma$ in the canonical model containing $\varphi$. Pick a distinct fresh object $v^\varphi_j$ of for each $\fdia \varphi \in L(u)$ and each $j \in \{1,...,\dfrl\}$, and define the graph $G'$ by taking its nodes to be the union of the set of nodes in $G$ together with all the fresh objects $v^\varphi_j$, and setting $w \edge{G'} w'$ iff either $w,w' \in G$ and $w \edge{G'} w'$, or $w = u$ and $w' = v^\varphi_j$ for some $\varphi,j$.   We set  $L'(v^\varphi_j) = \Theta_\varphi \cap \Sigma$ for $\fdia \varphi \in L(u)$ and $j \in \{1,...,\dfrl\}$, and for $w \in G$ we set $L'(w) = L(w)$. Finally, we set $S_P' = S_P$, and $S_F' = S_F \cup \{u\}$. As before, $\mathcal{N}'$ is a network and $\mathcal{N} \sqsubseteq \mathcal{N}'$, and $u$ is not an $\fdia$-defect in $\mathcal{N}'$. Furthermore, since the only new nodes we've added in $\mathcal{N}'$ are head nodes, and since $\mathcal{N} \sqsubseteq \mathcal{N}'$ and hence every node in $\mathcal{N}$ that is not an $\fdia$-defect is  not an $\fdia$-defect in $\mathcal{N}'$ either, the number of $\fdia$-defects in $\mathcal{N}'$ that are not head nodes is strictly smaller than the number of $\fdia$-defects in $\mathcal{N}$ that are not head nodes as required. 
\end{proof}

\begin{prop}
\label{p:singledef}
Let $\Sigma$ be a Fischer-Ladner closed finite subset of $\languagesym_\Delta(\Prop)$ where $\Delta$ is a set of guarded and disjunctive fixpoint connectives. 
Let $\mathcal{N}$ be a finite  and anticonfluent $\Sigma$-network, and let $d$ be a $\sharp_\chi \vec{\theta}$-deferral at some node $u$ in $\mathcal{N}$. Then there exists a finite  and anticonfluent network $\mathcal{N}' $ such that $\mathcal{N} \sqsubseteq \mathcal{N}'$ and such that $d$ is eventually finished at $u$ in $\mathcal{N}'$. 
\end{prop}

\begin{proof}
For forward-looking connectives we prove the result by induction on the greatest distance of the node $u$ to a head node. More precisely, the induction is on the greatest natural number $k$ that is equal to the length of a path: $$u \edge{G} v_1 \edge{G} ... \edge{G} v_n$$ such that $v_n$ is a head node, if $u$ is not itself a head node, and $k = 0$ otherwise. This is well defined because $\mathcal{N}$ is finite and anticonfluent, and furthermore since $\mathcal{N}$ is anti-confluent any successor of $u$ has shorter greatest distance to a head node than $u$.  Similarly for backward-looking connectives we prove the result by induction on the maximal length of a path to the node $u$ from a tail node. 

The precise induction hypothesis in the forward-looking case is: for each deferral $d$ at a node $u$ in $\mathcal{N}$, there exists a network $\mathcal{N}'$ such that: 
\begin{itemize}
\item $\mathcal{N} \sqsubseteq \mathcal{N}'$, 
\item $\mathcal{N}$ is downwards cofinal in $\mathcal{N}'$,
\item $d$ is eventually finished at $u$ in $\mathcal{N}'$, 
\item $\mathcal{N} \equp{u} \mathcal{N}'$.
\end{itemize}
Note that we can assume w.l.o.g. that the only $\fdia$-defects in $\mathcal{N}$ are head nodes, since otherwise we can apply Proposition \ref{p:onlyhead} to extend $\mathcal{N}$ to a finite and anticonfluent  network with this property, that has $\mathcal{N}$ as a subnetwork. 

Similarly, the induction hypothesis in the backward-looking case is: for each deferral $d$ at a node $u$ in $\mathcal{N}$, there exists a network $\mathcal{N}'$ such that: 
\begin{itemize}
\item $\mathcal{N} \sqsubseteq \mathcal{N}'$, 
\item $\mathcal{N}$ is upwards cofinal in $\mathcal{N}'$,
\item $d$ is eventually finished at $u$ in $\mathcal{N}'$, 
\item $\mathcal{N} \eqdown{u} \mathcal{N}'$.
\end{itemize}
Note that we can assume in this case w.l.o.g. that the only $\bdia$-defects in $\mathcal{N}$ are tail nodes, since otherwise we can apply Proposition \ref{p:onlyhead} to extend $\mathcal{N}$ to a finite and anticonfluent  network with this property, that has $\mathcal{N}$ as a subnetwork. 

We focus on the forward-looking case since the other case is proved in an entirely analogous manner. 
In the base case of distance $0$, if $u$ is a head node with $u \notin S_F$ the result follows from Proposition \ref{p:stronginductionlemma}. To be precise: by Proposition \ref{p:stronginductionlemma} there exists a finite network $\mathcal{N}' = (G',L',S_F',S_P')$ such that $G'$ is a tree rooted at a node $u'$ with $L'(u') = L(u)$, $S_P' = \emptyset$ and  such that the deferral $d$ is eventually finished at $u'$ in $\mathcal{N}'$. We assume w.l.o.g. that in fact $u = u'$ and $G \setminus \{u\} \cap G' = \emptyset$. We construct a new network $\mathcal{N}'' = (G'',L'')$ by taking $\underline{G}'' = \underline{G} \cup \underline{G'}$, and set $v \edge{G''} v'$ iff either $v,v' \in \underline{G}$ and $v \edge{G} v'$, or $v,v' \in \underline{G}'$ and $v \edge{G'} v'$. The labelling function $L''$ is defined as the union of $L$ and $L'$. It is easy to see that $G''$ is finite and anticonfluent, and since $u \notin S_F$ we have $\mathcal{N} \sqsubseteq \mathcal{N}''$. Also note that $\mathcal{N}$ is clearly downwards cofinal in $\mathcal{N}''$ and that $\mathcal{N} \equp{u} \mathcal{N}''$.  Furthermore, since $S_P' = \emptyset$ we also have $\mathcal{N}' \sqsubseteq \mathcal{N}''$, and so it follows from Proposition \ref{p:stayfinished} that the deferral $d$ is eventually finished at $u$ in the network $\mathcal{N}''$.

In the case of a head node $u \in S_F$, note that we must have $\fbox \bot \in L(u)$ since $u$ has no successors but is in $S_F$, and $\mathcal{N}$ is a network  -- so $\fdia \top \in L(u)$ would require $u$ to have at least one successor. (This little detail is in fact why we required $\fbox \bot, \bbox \bot \in \Sigma$ in the definition of Fischer-Ladner closure). We claim that the deferral $d$ is already eventually finished in $\mathcal{N}$. To prove this, we first show that for any deferral at $u$ of the form $\psi(x,\vec{q})$ such that every occurrence of $x$ in $\psi(x,\vec{q})$ is guarded, we have: 
$$\fbox \bot \wedge \psi(\sharp_\chi \vec{\theta},\vec{\theta}) \vdash \psi(\bot,\vec{\theta})$$  We prove this by induction on the complexity of the deferral $d$, treating formulas of the form \\ $\fnabla \{\psi_1(x,\vec{q}),...,\psi_m(x,\vec{q})\}$ as the base case of the induction. For this case, since we are assuming that $\fnabla \{\psi_1(x,\vec{q}),...,\psi_m(x,\vec{q})\}$ is a deferral and so must contain some occurrence of the variable $x$, we must have $\{\psi_1(x,\vec{q}),...,\psi_m(x,\vec{q})\} \neq \emptyset$, hence clearly $$\fnabla \{\psi_1(\sharp_\chi \vec{\theta},\vec{\theta}),...,\psi_m(\sharp_\chi \vec{\theta},\vec{\theta})\} \vdash \fdia \top$$ So we get:
\begin{eqnarray*}
 \fbox \bot \wedge \fnabla \{\psi_1(\sharp_\chi \vec{\theta},\vec{\theta}),...,\psi_m(\sharp_\chi \vec{\theta},\vec{\theta})\} & \vdash & \fbox \bot \wedge \fdia \top \\
& \vdash & \bot \\
& \vdash & \fnabla \{\psi_1(\bot,\vec{\theta}),...,\psi_m(\bot,\vec{\theta})\}
\end{eqnarray*}
as required. 

Supposing that the induction hypothesis holds for the deferral $\psi(x,\vec{q})$, it is clear that the claim holds also for the deferral $\gamma \wedge \psi(x,\vec{q})$ where $x$ does not appear in $\gamma$: we have $\fbox \bot \wedge \psi(\sharp_\chi \vec{\theta},\vec{\theta}) \vdash \psi(\bot,\vec{\theta})$ and so
$$\fbox \bot \wedge \gamma \wedge \psi(\sharp_\chi \vec{\theta},\vec{\theta}) \vdash \gamma \wedge \psi(\bot,\vec{\theta})$$
as required. Lastly, if the induction hypothesis holds for the formulas $\psi_1(x,\vec{q})$ and $\psi_2(x,\vec{q})$ and these are both deferrals, then we have  $\fbox \bot \wedge \psi_1(\sharp_\chi \vec{\theta},\vec{\theta}) \vdash \psi_1(\bot,\vec{\theta})$ and $\fbox \bot \wedge \psi_2(\sharp_\chi \vec{\theta},\vec{\theta}) \vdash \psi_1(\bot,\vec{\theta})$, so we get
  $$\fbox \bot \wedge (\psi_1(\sharp_\chi \vec{\theta},\vec{\theta}) \vee \psi_2(\sharp_\chi \vec{\theta},\vec{\theta})) \vdash \psi_1(\bot,\vec{\theta})$$ 
by elementary propositional logic. The case where only one of the formulas $\psi_1(x,\vec{q})$, $\psi_2(x,\vec{q})$ is a potential deferral is equally simple, and left to the reader.

We now prove, by induction on the complexity of a deferral $d$, that $d$ is eventually finished at $u$ in $\mathcal{N}$.  For $d = x$, we have $\sharp_\chi \vec{\theta} \in L(u)$ and so $\chi(\sharp_\chi \vec{\theta},\vec{\theta}) \in L(u)$, so $\chi(x,\vec{q})$ is a $\sharp_\chi\vec{\theta}$-deferral at $u$ in which the variable $x$ appears only in guarded positions. Hence we get, by the previous claim, that $\fbox \bot \wedge \chi(\sharp_\chi \vec{\theta},\vec{\theta}) \vdash \chi(\bot,\vec{\theta})$, so $\chi(\bot,\vec{\theta}) \in L(u)$ and so the deferral $x$ is indeed eventually finished at $u$. If  the induction hypothesis holds for the deferral $\psi(x,\vec{q})$, then it holds also for the deferral $\gamma \wedge \psi(x,\vec{q})$ where $x$ does not appear in $\gamma$, since if $\gamma \wedge \psi(\sharp_\chi \vec{\theta},\vec{\theta}) \in L(u)$ then $\gamma \in L(u)$ and $ \psi(\sharp_\chi \vec{\theta},\vec{\theta}) \in L(u)$, so by the induction hypothesis the deferral $\psi(x,\vec{q})$ is eventually finished at $u$ and hence so is $\gamma \wedge \psi(x,\vec{q})$. The induction step for a deferral of the form $\psi_1(x,\vec{q}) \vee \psi_2(x,\vec{q})$ follows  since if $\psi_1(\sharp_\chi \vec{\theta},\vec{\theta}) \vee \psi_2(\sharp_\chi \vec{\theta},\vec{\theta}) \in L(u)$ then for some $i \in \{1,2\}$ we have $\psi_i(\sharp_\chi \vec{\theta},\vec{\theta}) \in L(u)$, and we can apply the induction hypothesis. Finally, the case for a deferral of the form $\fnabla \{\psi_1(x,\vec{q}),...,\psi_m(x,\vec{q})\}$ holds vacuously, since $\fbox \bot \in L(u)$ and so we can never have $\fnabla \{\psi_1(\sharp_\chi \vec{\theta},\vec{\theta}),...,\psi_m(\sharp_\chi \vec{\theta},\vec{\theta})\} \in L(u)$: just as before we have $\{\psi_1(x,\vec{q}),...,\psi_m(x,\vec{q})\} \neq \emptyset$ and we can show with the same argument that $\fnabla \{\psi_1(\sharp_\chi \vec{\theta},\vec{\theta}),...,\psi_m(\sharp_\chi \vec{\theta},\vec{\theta})\}$ is inconsistent with $\fbox \bot$.

Now let $u$ be an arbitrary node with maximal distance $k > 0$ to a head node, and suppose that the induction hypothesis holds for all nodes $u'$ of maximal distance less than $k$ to a head node. We shall first prove the following: 

\begin{claim}
\label{c:guarded}
For every deferral $d$ at $u$ in which every occurrence of the variable $x$ is guarded,  there exists a network $\mathcal{N}'$ such that $\mathcal{N} \sqsubseteq \mathcal{N}'$,
 $\mathcal{N}$ is downwards cofinal in $\mathcal{N}'$,
 $d$ is eventually finished at $u$ in $\mathcal{N}'$, 
and $\mathcal{N} \equp{u} \mathcal{N}'$.
\end{claim}
\begin{proof}[Proof of Claim \ref{c:guarded}] We prove this by induction on the complexity of $\psi(x,\vec{q})$, where we treat modal formulas as the base case of the induction.  

We start with the case of $d = \fnabla \{\psi_1(x,\vec{q}),...,\psi_m(x,\vec{q})\}$ (base case of the induction). Since $u$ is not a head node, it is not an $\fdia$-defect, and so $u \in S_F$. So for each formula $\fdia \varphi \in L(u)$ and each $j \in \{1,...,\dfrl\}$, there exists a successor node $v_j^\varphi$ of $u$ such that $\varphi \in L(v_j^\varphi)$, $L(v_1^\varphi) = ... = L(v_\dfrl^\varphi)$ and $v_i^\varphi \neq v_j^\rho$ whenever $i \neq j$ or $\varphi \neq \rho$.  In particular, for each $j \in \{1,...,m\}$ and $i \in \{1,...,\dfrl\}$, there corresponds a successor $v_i^{\psi_j(\sharp_\chi \vec{\theta},\vec{\theta})}$.  We define a surjective map $f : W \to  \{\psi_1(x,\vec{q}),...,\psi_m(x,\vec{q})\}$, where $W$ is the set of successors of $u$, such that $\psi_i(\sharp_\chi \vec{\theta},\vec{\theta}) \in L(w)$ whenever $f(w) = \psi_i(x,\vec{q})$, as follows:
\begin{itemize}
\item If $w$ is of the form $v_i^{\psi_j(\sharp_\chi \vec{\theta},\vec{\theta})}$ for $i \in \{1,...,\dfrl\}$ such that $\eta(i) = \psi_j(x,\vec{q})$, then set $f(w) = \psi_j(x,\vec{q})$.
\item Otherwise, set $f(w)$ to be some arbitrarily chosen formula $\psi_j(x,\vec{q})$ such that $\psi_j(\sharp_\chi \vec{\theta},\vec{\theta}) \in L(w)$   -- this must exist since $w$ is a successor of $u$ and $\fbox (\psi_1(\sharp_\chi \vec{\theta},\vec{\theta}) \vee ... \vee\psi_m(\sharp_\chi \vec{\theta},\vec{\theta})) \in L(u)$.
\end{itemize} 
Let $Z$ be the graph of $f$. Let $W' \subseteq W$ denote the set of elements $w \in W$ such that $f(w)$ is a potential deferral. For each $w \in W'$ we can apply the induction hypothesis on the maximal distance from $w$ to a head node, and find a finite, anticonfluent network $\mathcal{N}_w$ with  $\mathcal{N} \sqsubseteq \mathcal{N}_w$, $\mathcal{N}_w \equp{w} \mathcal{N}$,  $\mathcal{N}$ is downwards cofinal in $\mathcal{N}_w$ and the deferral $f(w)$ is eventually finished at $w$ in $\mathcal{N}_w$. Also, since $\mathcal{N}$ is anticonfluent and the members of $W'$ are all successors of $u$, we can assume that the sets $\upgen{\mathcal{N}_w}{w}$ are pairwise disjoint: this holds since the sets $\upgen{\mathcal{N}_w}{w} \cap \underline{G}$ for $w \in W'$ are pairwise disjoint because $\mathcal{N}$ was anticonfluent, and we can just replace the parts of each network $\mathcal{N}_w$ outside of $\mathcal{N}$ with isomorphic copies. We can thus apply Proposition \ref{p:amalgamation} to the set of networks $\{\mathcal{N}_w \mid w \in W'\}$ and find a single network $\mathcal{N}' $ with $\mathcal{N} \sqsubseteq \mathcal{N}'$, $\mathcal{N}_w \sqsubseteq \mathcal{N}'$ for each $w$ such that $f(w)$ is a potential deferral (hence the deferral $f(w)$ is eventually finished at $w$ in $\mathcal{N}'$ by Proposition \ref{p:stayfinished}), $\mathcal{N} \equp{S'} \mathcal{N}'$,  and $\mathcal{N}$ is downwards cofinal in $\mathcal{N}'$. It clearly follows that $\mathcal{N} \equp{u} \mathcal{N}'$, since the elements of $W'$ are all successors of $u$. Finally, the deferral $\fnabla \{\psi_1(x,\vec{q}),...,\psi_m(x,\vec{q})\}$ is eventually finished at $u$ in $\mathcal{N}'$, witnessed by the relation $Z$.

If $d = \psi_1(x,\vec{q}) \vee \psi_2(x,\vec{q})$ and the induction hypothesis holds for $\psi_1(x,\vec{q})$ and $\psi_2(x,\vec{q})$ then $\psi_i(\sharp_\chi \vec{\theta},\vec{\theta}) \in L(u)$ for some $i \in \{1,2\}$. If $\psi_i(x,\vec{q})$ is not a deferral then we are done, otherwise we apply the induction hypothesis to find a finite anticonfluent network $\mathcal{N}'$ with $\mathcal{N} \sqsubseteq \mathcal{N}'$ satisfying all the constraints of the induction hypothesis and such that the deferral $\psi_i(x,\vec{q})$ is eventually finished at $u$ in $\mathcal{N}'$. Hence the deferral $d$ is also eventually finished at $u$ in $\mathcal{N}'$.

If $d = \gamma \wedge \psi(x,\vec{q})$ where $x$ does not appear in  $\gamma$ then we have $\gamma \in L(u)$ and $\psi(\sharp_\chi \vec{\theta},\vec{\theta}) \in L(u)$, so if the induction hypothesis holds for $\psi(x,\vec{q})$ then there is a finite anticonfluent network $\mathcal{N}'$ with $\mathcal{N} \sqsubseteq \mathcal{N}'$, satisfying all the constraints of the induction hypothesis, and such that the deferral $\psi_(x,\vec{q})$ is eventually finished at $u$ in $\mathcal{N}'$. So the deferral $ \gamma \wedge \psi(x,\vec{q})$ is also eventually finished at $u$ in $\mathcal{N}'$.
 This completes the induction and thus the proof of the Claim.
\end{proof}
With Claim \ref{c:guarded} in place, we finish the proof by induction on the complexity of the deferral $d$ (in which $x$ is not necessarily guarded). If $d = x$, then $\chi(\sharp_\chi \vec{\theta},\vec{\theta}) \in L(u)$, so $\chi(x,\vec{q})$ is a $\sharp_\chi \vec{\theta}$-deferral at $u$. Since every occurrence of $x$ in the formula $\chi(x,\vec{q})$ is guarded (because we assumed $\Delta$ was a set of guarded fixpoint connectives), Claim \ref{c:guarded} applies and so there exists a suitable finite and anticonfluent network $\mathcal{N}'$ extending $\mathcal{N}$ and satisfying all the requirements of the induction hypothesis, in which the deferral $\chi(x,\vec{q})$ is eventually finished at $u$. Hence the deferral $x$ is also eventually finished at $u$ in $\mathcal{N}'$, as required. The rest of the induction follows precisely the same reasoning as in the proof of Claim \ref{c:guarded}, so we omit it.
\end{proof}

\begin{prop}
\label{p:alldef}
Let $\mathcal{N}$ be a finite and anticonfluent network. Then there exists a  finite and anticonfluent network $\mathcal{N}' $ such that $\mathcal{N} \sqsubseteq \mathcal{N}'$ and such that none of the nodes belonging to $\mathcal{N}$ are $\fdia$-defects or $\bdia$-defects in $\mathcal{N}'$, and none of the nodes in $\mathcal{N}$ have any $\mu$-defects in $\mathcal{N}'$.
\end{prop}

\begin{proof}
Just list all the finitely many defects $d_1,...,d_k$ appearing in $\mathcal{N}$, and repeatedly apply Propositions \ref{p:removefpdefects} and \ref{p:singledef} to construct a chain of networks $\mathcal{N} \sqsubseteq \mathcal{N}_1 \sqsubseteq ... \sqsubseteq \mathcal{N}_k$ removing these defects one by one (and relying on Proposition \ref{p:stayfinished} to make sure that the defects removed at some stage are not reintroduced in later stages). None of the nodes  in the network $\mathcal{N}_k$ can then be $\fdia$-defects or $\bdia$-defects, and none of the $\mu$-defects left in $\mathcal{N}_k$ can belong to nodes in $\mathcal{N}$.
\end{proof}

\begin{prop}
\label{p:perfect}
Any $\Sigma$-atom $A$ belongs to the label of some node in a perfect network.
\end{prop}

\begin{proof}
First, we define a network $\mathcal{N}_0 = (G_0,L_0,S_F^0,S_P^0)$ by setting $G_0$ to be a singleton $u$ with empty set of edges, $L_0(u) = A$, $S_F^0 = S_P^0 = \emptyset$. Given that we have defined the network $\mathcal{N}_i$ for $i \in \omega$, let $\mathcal{N}_{i + 1}$ be a network with $\mathcal{N}_i \sqsubseteq \mathcal{N}_{i + 1}$ and such that none of the nodes in $\mathcal{N}_i$ are  $\fdia$-defects, or $\bdia$-defects, or contain any $\mu$-defects, in $\mathcal{N}_{i + 1}$. Such an extension is guaranteed to exist by Proposition \ref{p:alldef}. By this inductive procedure we find an infinite chain:
$$\mathcal{N}_0 \sqsubseteq \mathcal{N}_1 \sqsubseteq \mathcal{N}_2 \sqsubseteq ...$$
Set $\mathcal{N}_\omega = \bigcup_{j \in \omega} \mathcal{N}_j$. It is not hard to see that $\mathcal{N}_i \sqsubseteq \mathcal{N}_\omega$ for all $i \in \omega$. So if a node $v$ in $\mathcal{N}_\omega$ is an $\fdia$-defect, or a $\bdia$-defect, or contains a $\mu$-defect, then by Proposition \ref{p:stayfinished} this defect must be present in $\mathcal{N}_{i + 1}$, where $i$ is the smallest index for which the node $u$ appears in $\mathcal{N}_i$ . But this is a contradiction by definition of $\mathcal{N}_{i + 1}$, so we are done. 
\end{proof}

We can now prove our main theorem:

\begin{proof}[Proof of Theorem \ref{t:main}]
Suppose $\nvdash \varphi$. Then $\neg \varphi$ is consistent. Let $\Sigma$ be the Fischer-Ladner closure of $\neg\varphi$ and let $\Gamma$ be a maximal consistent set containing $\neg\varphi$. By Proposition \ref{p:perfect} there is a perfect $\Sigma$-network $\mathcal{N}$ and a node $u$ in $\mathcal{N}$ labelled $\Gamma \cap \Sigma$. By Proposition \ref{p:ismodel} we thus get $\bbS_\mathcal{N},u \Vdash \neg \varphi$, so $\nVdash \varphi$ as required.
\end{proof}

\section{Removing the guards}
\label{guardedness}
Finally, we show how to remove the assumption of guardedness used in Theorem \ref{t:main}.

\begin{prop}
\label{p:guardification}
For each disjunctive fixpoint connective $\chi$, there exists a guarded disjunctive fixpoint connective $\gamma$ such that:
$$\logicsym_{\{\chi,\gamma\}} \vdash \sharp_\chi \vec{z}\leftrightarrow  \sharp_\gamma \vec{z}$$
\end{prop}

\begin{proof}
By a standard argument, the formula $\chi(x,\vec{q})$ can be rewritten equivalently (with respect to provable equivalence in the system $\logicsym$) as a formula of the form:
$$(x \wedge \gamma_1(x,\vec{q})) \vee \gamma_2(x,\vec{q})$$
where $x$ is guarded in $\gamma_1,\gamma_2$ and $\gamma_1,\gamma_2$ are disjunctive. It is a routine exercise in fixpoint logic to show that: $$\logicsym_{\{\chi,\gamma\}} \vdash \sharp_\chi \vec{z}\leftrightarrow  \sharp_{\gamma_2}\vec{z}.$$
See \cite{enqv:comp16a} for more details.
\end{proof}

\begin{defi}
A set of fixpoint connectives $\Delta$ is said to be a \emph{guardification} of $\Gamma$ if:
$$\Delta = \{\gamma_2(x,\vec{q}) \mid \gamma(x,\vec{q}) \in \Gamma\}$$
where, for each $\gamma \in \Gamma$, $\gamma_1(x,\vec{q}),\gamma_2(x,\vec{q})$ are formulas in which $x$ is guarded, such that 
$$\logicsym \vdash \gamma(x,\vec{q}) \leftrightarrow (x \wedge \gamma_1(x,\vec{q})) \vee \gamma_2(x,\vec{q})$$
\end{defi}

\begin{prop}
\label{p:gconservative}
Let $\Delta$ be a guardification of $\Gamma$. Then the logic $\logicsym_{\Gamma \cup \Delta}$ is a conservative extension of $\logicsym_{\Gamma}$. 
\end{prop}

\begin{proof}
We claim that every $\logicsym_\Gamma$-algebra is also a $\logicsym_{\Gamma \cup \Delta}$-algebra, from which the result follows by algebraic completeness of both logics. 
To prove the claim, clearly the map $\chi_\alg(-,\vec{b})$ has a least pre-fixpoint if $\chi(x,\vec{q}) \in \Gamma$ since $\alg$ was a $\logicsym_\Gamma$-algebra. For $\gamma_2(x,\vec{q}) \in \Delta$, pick some $\chi(x,\vec{q}) \in \Gamma$ such that $\logicsym \vdash \chi(x,\vec{q}) \leftrightarrow (x \wedge \gamma_1(x,\vec{q}) \vee \gamma_2(x,\vec{q})$ and pick a tuple $\vec{b} \in A^n$. Since $\alg$ was a $\logicsym_\Gamma$-algebra, the map $\chi_\alg(-,\vec{b})$ has a least pre-fixpoint in $\alg$, call it $l$. We claim that $l$ is a least pre-fixpoint of the map $(\gamma_2)_\alg(-,\vec{b})$ as well. 

To see that $l$ is a pre-fixpoint, it suffices to note that $\chi_\alg(l,\vec{b}) \leq l$ implies $(\gamma_2)_\alg(l,\vec{b}) \leq l$ since $(\gamma_2)_\alg(l,\vec{b}) \leq  \chi_\alg(l,\vec{b})$.
To show that $l$ is the least pre-fixpoint, let $a$ be any pre-fixpoint for $(\gamma_2)_\alg(-,\vec{b})$. To show that $l \leq a$, we only need to show that $a$ is a pre-fixpoint for the map $\chi_\alg(-,\vec{b})$ as well. We know that: $$\logicsym \vdash \chi(x,\vec{q}) \leftrightarrow (x \wedge \gamma_1(x,\vec{q})) \vee \gamma_2(x,\vec{q})$$ 
and so, since $\alg$ was a $\logicsym$-algebra, we get:
\begin{eqnarray*}
\chi_\alg(a,\vec{b}) & = &  (a \wedge_\alg (\gamma_1)_\alg(a,\vec{b})) \vee_\alg (\gamma_2)_\alg(a,\vec{b}) \\
& \leq &  a \vee_\alg (\gamma_2)_\alg(a,\vec{b}) \\
& \leq &  a \vee_\alg a \\
& =  & a
\end{eqnarray*}
as required. 
\end{proof}

\begin{prop}
\label{p:gtransfercompleteness}
Suppose that $\Gamma$ is a set of fixpoint connectives, and suppose that $\Delta$ is a guardification of $\Gamma$. If $\logicsym_\Delta$ is Kripke complete, then so is $\logicsym_\Gamma$.
\end{prop}

\begin{proof}
Suppose the formula $\varphi \in \languagesym_\Gamma$ is valid on all Kripke frames.  Let $\Delta$ be a guardification of $\Gamma$, and for each fixpoint connective $\chi \in \Gamma$, let $t(\chi)$ be some  connective in $\Delta$ such that $\logicsym_{\Gamma \cup \Delta} \vdash \sharp_\chi \vec{z} \leftrightarrow \sharp_{t(\chi)} \vec{z}$.  We define a translation $t : \languagesym_\Gamma \to \languagesym_\Delta$ as follows: for a purely boolean formula $\pi$ (i.e. with no occurrences of modal or fixpoint operators) we set $t(\pi) = \pi$, and we let the translation $t$ commute with boolean connectives and modal operators. Finally, we set $t(\sharp_\chi (\theta_1,...,\theta_n)) = \sharp_{t(\chi)}(t(\theta_1),...,t(\theta_n))$. A straightforward induction  on $\varphi$, using soundness of the system $\logicsym_{\Gamma \cup \Delta}$ for the fixpoint case, shows that $\varphi$ and $t(\varphi)$ are semantically equivalent (w.r.t. the Kripke semantics), so $t(\varphi)$ is also valid on all Kripke frames.  By completeness of $\logicsym_\Delta$, we get $\logicsym_\Delta \vdash t(\varphi)$, and so clearly $\logicsym_{\Gamma \cup \Delta} \vdash t(\varphi)$.

We claim that:
$$\logicsym_{\Gamma \cup \Delta} \vdash \varphi \leftrightarrow t(\varphi)$$
If we can prove this we are done, since we then get $\logicsym_{\Gamma \cup \Delta} \vdash \varphi $ and so $\logicsym_{\Gamma} \vdash \varphi $ by Proposition \ref{p:gconservative}. 

The claim is proved by induction on the complexity of the formula $\varphi$. The purely propositional case is trivial, and modal operators and boolean connectives are both handled by just unfolding the definition of the map $t$ and applying the induction hypothesis. For $\varphi = \sharp_\chi (\theta_1,...,\theta_n)$, where $\chi \in \Gamma$, we have $\logicsym_{\Gamma \cup \Delta} \vdash \sharp_\chi \vec{z} \leftrightarrow \sharp_{t(\chi)} \vec{z}$ by definition of $t$, and  by the induction hypothesis on $\theta_1,...,\theta_n$ we have $\logicsym_{\Gamma \cup \Delta}\vdash  \theta_i \leftrightarrow t(\theta_i)$ for $1 \leq i \leq n$. So we get  
$$\logicsym_{\Gamma \cup \Delta} \vdash \sharp_\chi (\theta_1,...,\theta_n) \leftrightarrow \sharp_{t(\chi)} (t(\theta_1),...,t(\theta_n))$$
and we are done since  $t(\sharp_\chi (\theta_1,...,\theta_n)) = \sharp_{t(\chi)}(t(\theta_1),...,t(\theta_n))$.
\end{proof}

Putting these results together, we get a completeness result for all disjunctive fixpoint connectives, guarded or no:

\begin{theo}
\label{t:extendedguarded}
Let $\Gamma$ be any set of disjunctive fixpoint connectives. Then the logic $\logicsym_\Gamma$ is Kripke-complete. 
\end{theo}

\begin{proof}
By Proposition \ref{p:guardification}, $\Gamma$ has a guardification $\Delta$, and by Theorem \ref{t:main} the logic $\logicsym_\Delta$ is Kripke complete. So by Proposition \ref{p:gtransfercompleteness}, $\logicsym_\Gamma$ is Kripke complete too.
\end{proof}

%
%\begin{proofof}{Theorem \ref{t:fullcompleteness}}
%Let $\Gamma$ be any set of uni-directional fixpoint connectives.
%By Proposition \ref{p:guardification}, $\Gamma$ has a guardification $\Delta$, and by Theorem \ref{t:extendedguarded} the logic $\logicsym_\Delta$ is Kripke complete. So by Proposition \ref{p:gtransfercompleteness}, $\logicsym_\Gamma$ is Kripke complete too.
%\end{proofof}

%\subsection{Completeness for arbitrary flat modal fixpoint logics}
%Without converse, apply completeness for untied connectives (easier than main result here) with conservative extension proof

\section{Future research}
We conclude by listing a few topics for future research:
\begin{itemize}
\item Can the  method for axiomatizing arbitrary flat fixpoint logics via a ``simulation'' by disjunctive systems of fixpoint equations introduced in \cite{sant:comp10} be adapted to this setting? This could be quite problematic, since we required our disjunctive formulas to be either forward or backward looking. It is clear that allowing fixpoint connectives with a mix of both modalties results in an increased expressive power compared to the logic with only forward and backwards looking connectives, so simulating arbitrary connectives by disjunctive ones should not be possible. The key issue to address here is to find a more general notion of ``disjunctive formula'', which is restrictive enough so that the completeness proof still goes through, but allows mixing of forward and backwards modalities.
\item Can we do without the constraint of ``disjunctiveness'' altogether, and just prove completeness for the Kozen axioms for arbitrary flat modal fixpoint logics? This question has still not been resolved even in the case of flat fixpoint logics with only forward-looking modalities as studied in \cite{sant:comp10}. Perhaps one could solve the problem by further exploiting the use of ``foci'' for model constructions.
\item With regards to the previous question, recent  work by Afshari and Leigh \cite{afsh:cut17}, in which completeness is proved for the $\mu$-calculus without using disjunctive normal forms, might provide clues towards a solution. However, their proof does make use of the fact that a consistent formula $\gamma \wedge \mu x. \varphi(x)$ can be ``strengthened'' to a consistent formula $\gamma \wedge \mu x. \varphi(\neg \gamma \wedge x)$, and flat fixpoint languages need not be closed under such strengthenings in general.
\item Our completeness proof implies that the free algebras of logics $\logicsym_\Gamma$, where $\Gamma$ is a disjunctive set of fixpoint connectives, are constructive. Is there a purely algebraic proof of this? We know that finitary $\cO$-adjoints will not work, but this does not entail that no argument by purely algebraic reasoning could.
\item Finally, we hope to continue this line of work to study the problem of axiomatizing enriched modal $\mu$-calculi and fragments of such logics more generally. In particular, a challenge would be to prove completeness of Kozen's axioms for the full two-way $\mu$-calculus and for the hybrid $\mu$-calculus. Related work here is \cite{tamu:smal13}, where an axiomatization of the hybrid $\mu$-calculus is provided, relying on a small model theorem. However, this is quite different from Kozen's axiom system, and the rule for fixpoints directly refers to the bound on the size of finite satisfying models. An immediate thing to check is whether our completeness proof  for flat fixpoint logics with converse can also be adapted to axiomatize flat hybrid fixpoint logics.
\end{itemize}

{
\bibliographystyle{plain}
\bibliography{mu,extra,ml-book,nabla}
}

\end{document}